\documentclass{article}
\usepackage[margin=1.0in]{geometry}

\usepackage{microtype}
\usepackage{graphicx}
\usepackage{subfigure}
\usepackage{booktabs} 

\usepackage[colorlinks=true,linkcolor=blue,urlcolor=blue,citecolor=blue]{hyperref}

\usepackage{amsmath,amssymb,amsfonts}
\usepackage{algorithmic}
\usepackage{graphicx}
\usepackage{textcomp}
\usepackage[dvipsnames]{xcolor}
\usepackage{bm}
\usepackage{multirow}

\usepackage{academicons}
\definecolor{orcidlogocol}{HTML}{A6CE39}

\def\1{\bm{1}}

\DeclareMathAlphabet{\mathsfit}{\encodingdefault}{\sfdefault}{m}{sl}
\SetMathAlphabet{\mathsfit}{bold}{\encodingdefault}{\sfdefault}{bx}{n}

\newcommand{\R}{\mathbb{R}}

\DeclareMathOperator*{\argmin}{arg\,min}

\usepackage{enumitem}
\usepackage[capitalise]{cleveref}
\crefformat{section}{\S#2#1#3}

\usepackage[utf8]{inputenc} 
\usepackage[T1]{fontenc}    
\usepackage{hyperref}       
\usepackage{url}            
\usepackage{booktabs}       
\usepackage{nicefrac}       
\usepackage{microtype}      
\usepackage{multirow}
\usepackage{algorithm}

\usepackage[export]{adjustbox}
\usepackage{graphbox}

\def\xstar{x^\star}
\def\xfp{x^{(\infty)}}

\def\T{\top}

\def\ie{{\em i.e.,~}}
\def\eg{{\em e.g.,~}}

\DeclareMathOperator{\prox}{prox}

\newcommand{\xhat}{\widehat{x}}
\newcommand{\x}[1]{x^{(#1)}}

\usepackage{amsthm}
\newtheorem{thm}{Theorem}

\begin{document}

\title{Deep Equilibrium Architectures for Inverse Problems in Imaging}

\author{Davis Gilton\thanks{D.\,Gilton is with the Department of Electrical and Computer Engineering at the University of Wisconsin-Madison, 1415 Engineering Dr, Madison, WI 53706 USA.}, Gregory Ongie\thanks{G.\,Ongie is with the Department of Mathematical and Statistical Sciences at Marquette University, 1250 W Wisconsin Ave, Milwaukee, WI 53233 USA.}, and Rebecca Willett\thanks{R.\,Willett is with the Departments of Computer Science and Statistics at the University of Chicago, 5747 S Ellis Ave, Chicago, IL 60637 USA.} 	\thanks{The authors gratefully acknowledge funding from NSF Awards DMS‐1925101, DMS‐2023109, and OAC‐1934637 and AFOSR FA9550‐18‐1‐0166.}
}

\maketitle

\begin{abstract}
Recent efforts on solving inverse problems in imaging via deep neural networks use architectures inspired by a fixed  number of iterations of an optimization method. The number of iterations is typically quite small due to difficulties in training networks corresponding to more iterations; the resulting solvers cannot be run for more iterations at test time without incurring significant errors. This paper describes an  alternative approach corresponding to an {\em infinite} number of iterations, yielding 
a consistent improvement
in reconstruction accuracy above state-of-the-art alternatives and where the computational budget can be selected at test time to optimize context-dependent trade-offs between accuracy and computation. The proposed approach leverages ideas from Deep Equilibrium Models, where the fixed-point iteration is constructed to incorporate a known forward model and insights from classical optimization-based reconstruction methods. 
\end{abstract}

\section{Introduction}

A collection of recent efforts surveyed in \cite{ongie2020deep} consider the problem of using training data to solve inverse problems in imaging. 
Specifically, imagine we observe a corrupted set of measurements $y$ of an image $\xstar$ under a linear measurement operator $A$ with some noise $\varepsilon$ according to
\begin{equation}
    y = A \xstar + \varepsilon.
    \label{eq:obs}
\end{equation}
Our task is to compute an estimate of $\xstar$ given measurements $y$ and knowledge of $A$. This task is particularly challenging when the inverse problem is \emph{ill-posed}, \ie when the system is underdetermined or ill-conditioned, in which case simple methods such as least squares estimation may not have a unique solution or may produce estimates that are highly sensitive to noise.

\begin{figure*}[ht!]
\renewcommand*{\arraystretch}{0}
\centering
\begin{tabular}{c@{}c@{}c@{}c@{}c@{}c@{}}
& K=1 & K=10 & K=20 & K=30 & K=40 \vspace{2pt} \\
\begin{minipage}{0.09\linewidth} {\sc DE-Prox} \\ (ours) \end{minipage} & \includegraphics[width = 0.18\linewidth, align=c]{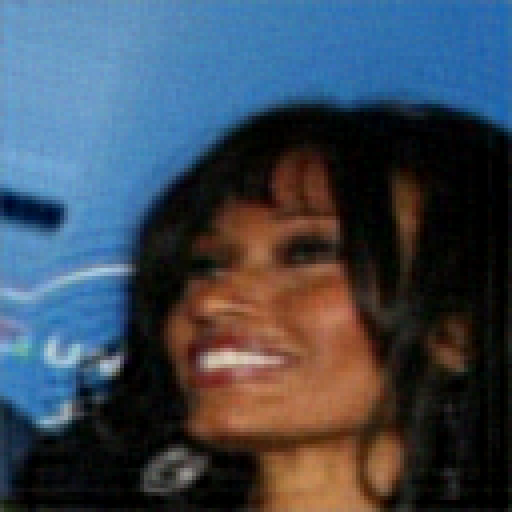} &
\includegraphics[width = 0.18\linewidth, align=c]{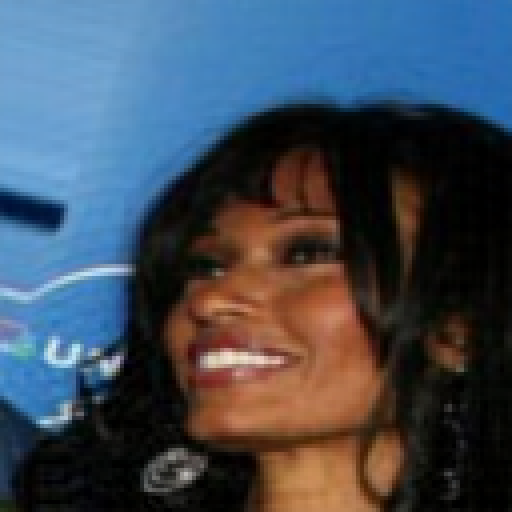} &
\includegraphics[width = 0.18\linewidth, align=c]{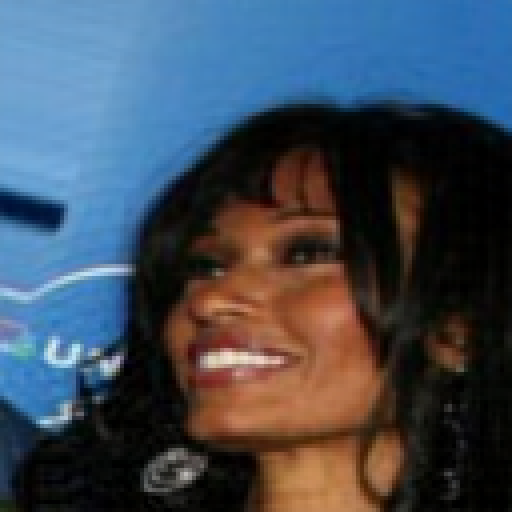} &
\includegraphics[width = 0.18\linewidth, align=c]{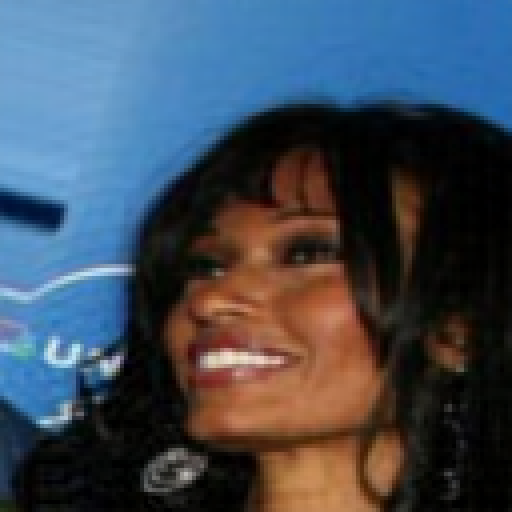} &
\includegraphics[width = 0.18\linewidth, align=c]{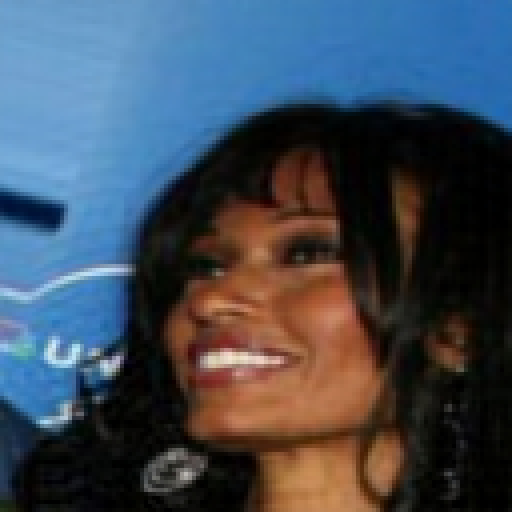} \\
\begin{minipage}{0.09\linewidth} {\sc DU-Prox}\end{minipage} & \includegraphics[width = 0.18\linewidth, align=c]{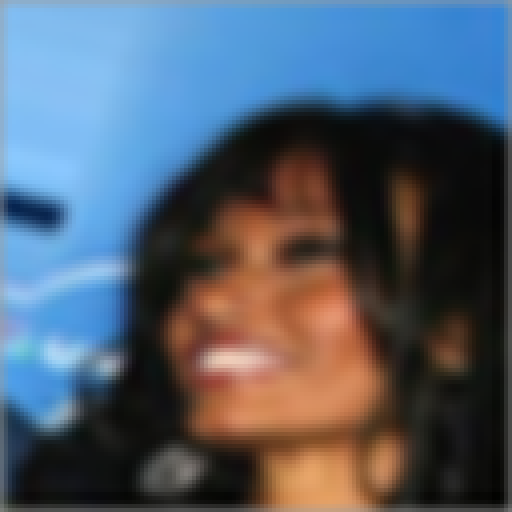} &
\includegraphics[width = 0.18\linewidth, align=c]{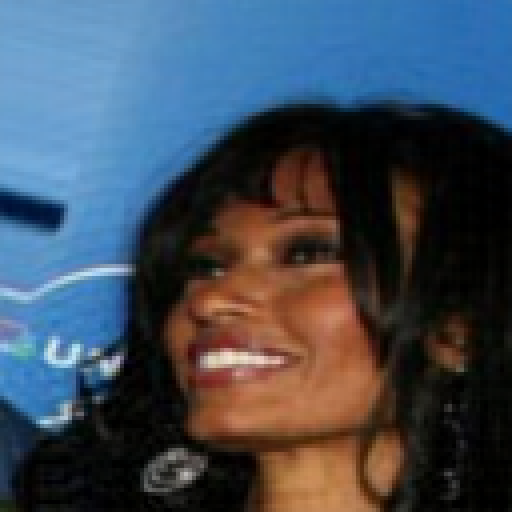} &
\includegraphics[width = 0.18\linewidth, align=c]{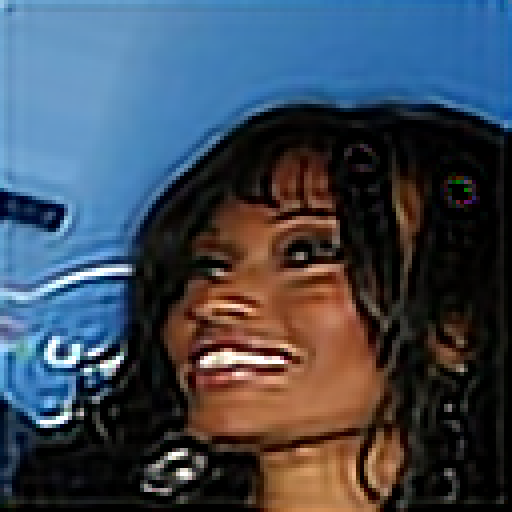} &
\includegraphics[width = 0.18\linewidth, align=c]{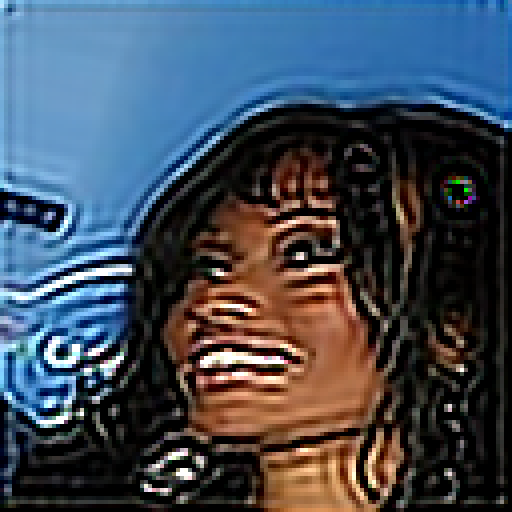} &
\includegraphics[width = 0.18\linewidth, align=c]{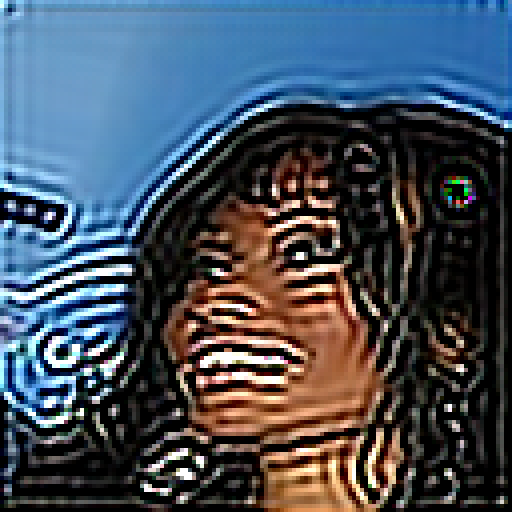} 
\end{tabular}
\caption{Deep Unrolling (DU) methods are state-of-the-art deep networks for image reconstruction that unroll iterative optimization algorithms for a fixed number of iterations $K$. As shown above in an illustrative example with Deep Unrolled Proximal Gradient Descent ({\sc DU-Prox}), these methods do not allow flexible operation at inference: unrolling for $K$ iterations where $K$ was not used at training results in severe artifacts. By utilizing Deep Equilibrium networks ({\sc DE-Prox} above), our method trains inverse solvers to return good reconstructions \emph{at convergence}, instead of at an arbitrary number of iterations, resulting in a flexible, higher-performing image reconstruction technique.}
\label{fig:debluriterates}
\end{figure*}

\begin{figure}[ht]
\begin{center}
\subfigure[MRI reconstruction.]{\includegraphics[width=.4\columnwidth]{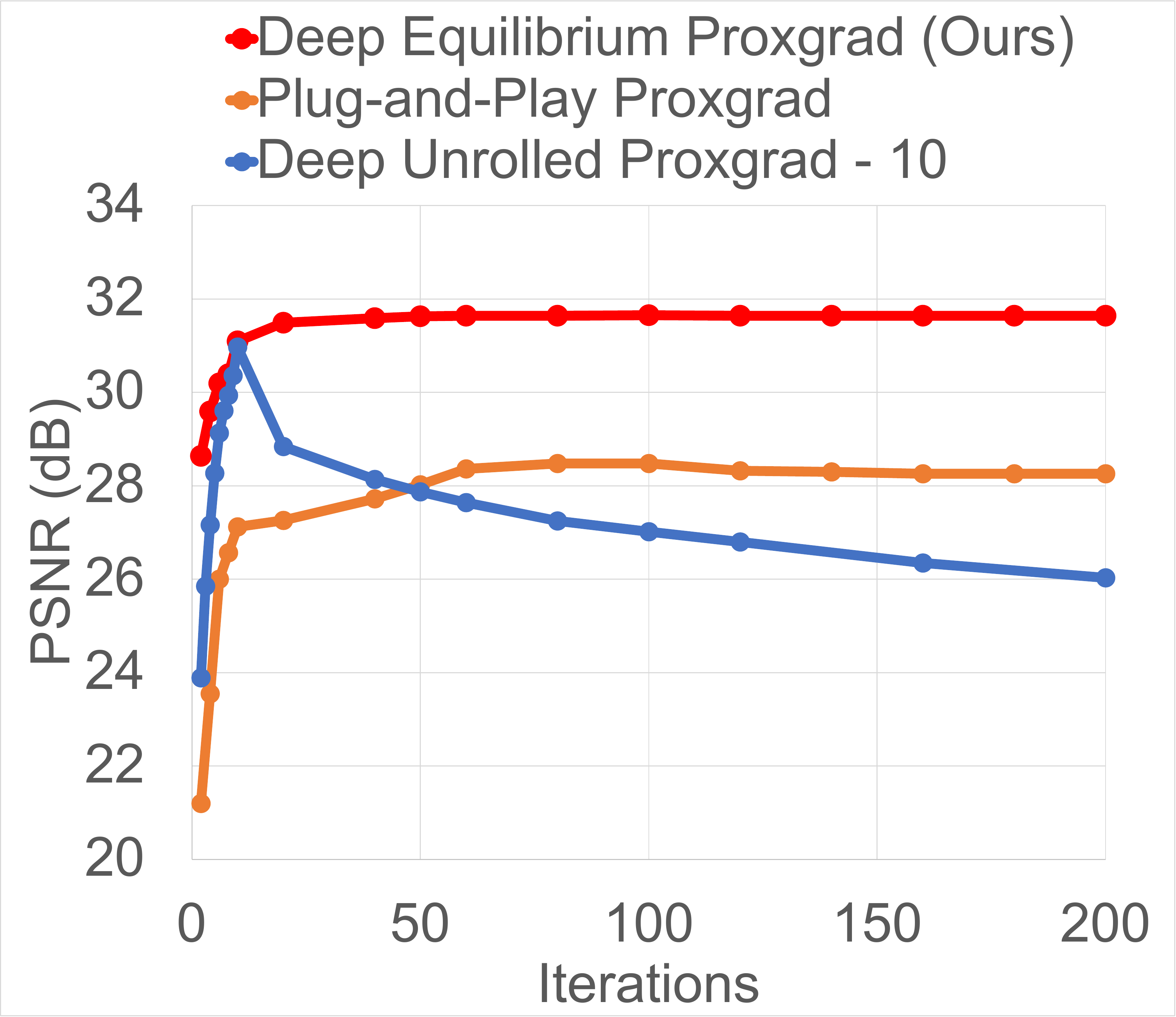} \label{fig:psnrclock}}~
\subfigure[Deep unrolling challenges]{\includegraphics[width=.4\columnwidth]{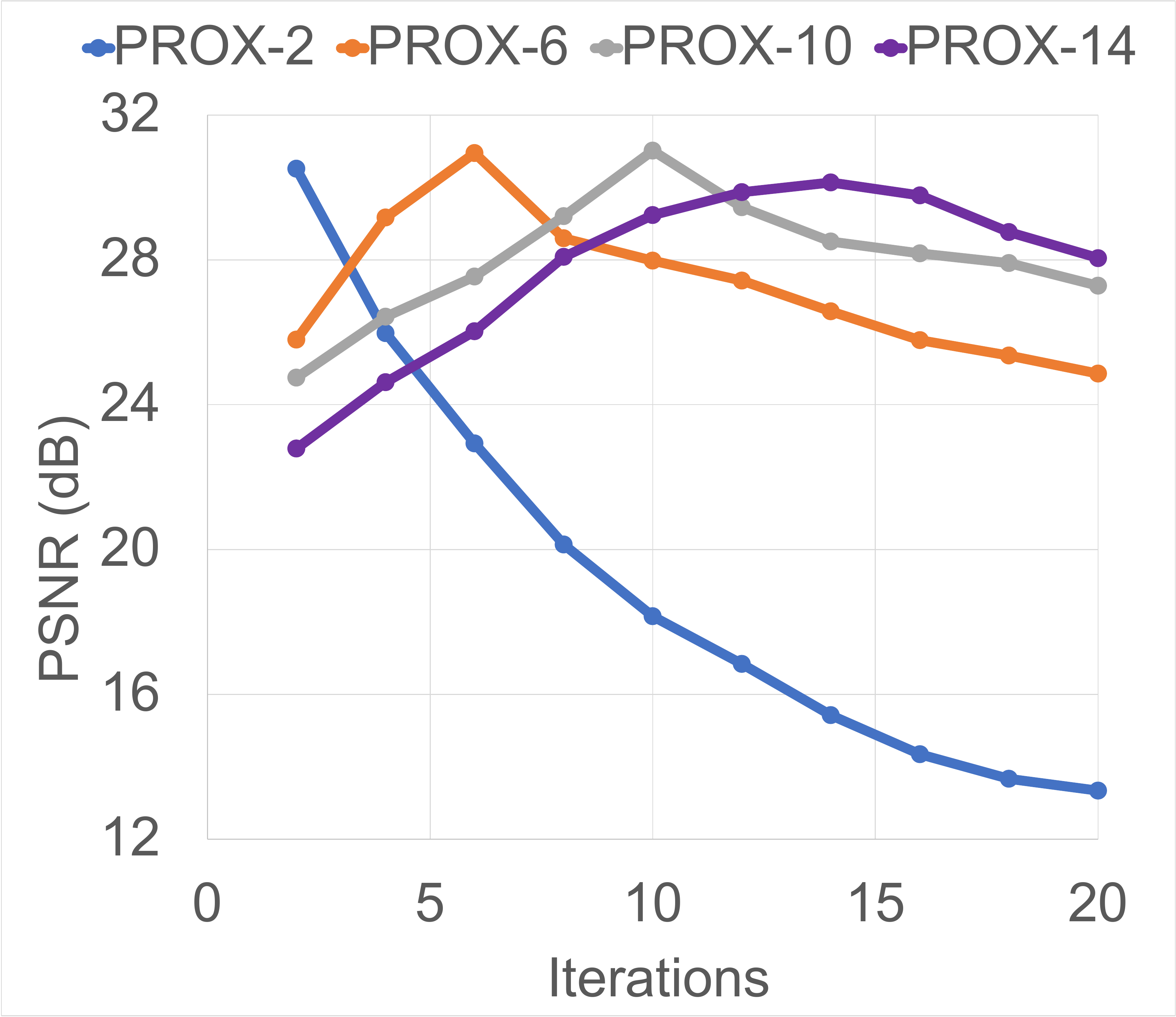}\label{fig:e2eiters}} 
\vspace{-.1in}
\caption{(a)
PSNR of reconstructed images for an MRI reconstruction problem as a function of iterations used to compute the reconstruction. Unrolled methods are optimized for a fixed computational budget during training, and running additional steps at inference yields a significant drop in performance. Our deep equilibrium methods can achieve the same PSNR for the optimal computational budget of an unrolled method, but can trade slightly more computation time for a significant increase in the PSNR, allowing a user to choose the desired computational budget and reconstruction quality. (b) Standard unrolled deep optimization networks typically require choosing some fixed number of iterates during training. Deviating from this fixed number at inference incurs a significant penalty in PSNR. The forward model here is 8x accelerated single-coil MRI reconstruction, and the unrolled algorithm is unrolled proximal gradient descent with $K$ iterates, labeled {\sc PROX-K} (Fig. \ref{fig:deprox}). For further experimental details see \cref{sec:results}.
}
\end{center}
\vskip -0.2in
\end{figure}

Decades of research has explored geometric models of image structure that can be used to regularize solutions to this inverse problem, including \cite{tikhonov1943stability,rudin1992nonlinear,dabov2007image} and many others. More recent efforts have focused instead on using large collections of training images, $\{x_i^*\}_{i=1}^n$, to {\em learn} effective regularizers.

One particularly popular and effective approach involves augmenting standard iterative inverse problem solvers with learned deep networks. This approach, which we refer to as {\bf deep unrolling (DU)}, is reviewed in \cref{sec:DU}. 
The basic idea is to build an architecture that mimics a small number of iterations of an iterative algorithm. In practice, the number of iterations is quite small (typically 5-10) because of issues stability, memory, and numerical issues arising in backpropagation. This paper sidesteps this key limitation of deep unrolling methods with a novel approach based on {\bf deep equilibrium models (DEMs)} \cite{bai2019deep}, which are designed for training arbitrarily deep networks. The result is a novel approach to training networks to solve inverse problems in imaging that {\bf yields a consistent improvement in performance above state-of-the-art alternatives} 
and where {\bf the computational budget can be selected at test time} to optimize context-dependent tradeoffs between accuracy and computation. The key empirical findings, which are detailed in \cref{sec:results}, are illustrated in Fig.~\ref{fig:psnrclock}.

\subsection{Contributions}

This paper presents a novel approach to machine learning-based methods for solving linear inverse problems in imaging. Unlike most state-of-the-art methods, which are based on unrolling a small number of iterations of an iterative reconstruction scheme (``deep unrolling''), our method is based on deep equilibrium models that correspond to a potentially infinite number of iterations. This framework yields more accurate reconstructions that the current state-of-the-art across a range of inverse problems and gives users the ability to navigate a tradeoff between reconstruction computation time and accuracy during inference. Furthermore, because our formulation is based on finding a fixed points of a operator, we can use standard acceleration techniques to speed inference computations -- something that is not possible with deep unrolling methods.
In addition, our approach 
inherits provable convergence guarantees depending on the ``base'' algorithm used to select a fixed point equation for the deep equilibrium framework. Experimental results also show that our proposed
initialization for Deep Equilibrium Models based on pre-training is superior to random initialization, and the proposed approach is more robust to noise than past methods.  
Overall, the proposed DEM approach is a unique bridge between conventional fixed-point methods in numerical analysis and learning-based techniques for inverse problems in imaging.

\section{Relationship to Prior Work}
\subsection{Review of Deep Unrolling Methods}\label{sec:DU}
\label{sec:unrolling}

Deep unrolling methods for solving inverse problems in imaging consist of a fixed number of architecturally identical ``blocks,'' which are often inspired by particular optimization algorithm. These methods represent the current state-of-the-art in MRI reconstruction, with most top submissions to the fastMRI challenge \cite{muckley2020state} being some sort of unrolled net. Deep unrolling architectures have also been successfully applied to other inverse problems in imaging, such as low-dose CT \cite{wu2019computationally}, light-field photography \cite{chun2020momentum}, and emission tomography \cite{mehranian2020model}.

We describe here a specific deep unrolling method based on the gradient descent algorithm, although many other variants exist based on alternative optimization or fixed point iteration schemes \cite{ongie2020deep}. Suppose we have a known regularization function $r$ that could be applied to an image $x$; \eg in Tikhonov regularization, $r(x) = \frac{\lambda}{2} \|x\|_2^2$ for some scalar $\lambda > 0$. Then we could compute an image estimate $\xhat$ by solving the optimization problem
\begin{equation}
\xhat = \argmin_x \frac{1}{2}\|y-Ax\|_2^2 + r(x). \label{eq:opt}
\end{equation}
If $r$ is differentiable, this can be accomplished via gradient descent. That is, we start with an initial estimate $\x{0}$ such as $\x{0} = A^\top y$ and choose a step size $\eta > 0$, such that for iteration $k = 1,2,3,\ldots$, we set
$$\x{k+1} = \x{k} + \eta A^\top(y-A\x{k}) - \eta \nabla r(\x{k}),$$
where $\nabla r$ is the gradient of the regularizer.

The basic idea behind deep unrolled methods is to fix some number of iterations $K$ (typically $K$ ranges from 5 to 10), declare that $\x{K}$ will be our estimate $\xhat$, and model $\nabla r$ with a neural network, denoted $R_\theta(x)$, whose weights $\theta$ can be learned from training data. 
For example, we may define the unrolled gradient descent estimate to be $\xhat^{(K)}(y;\theta) := x^{(K)}$ where $x^{(0)} = A^\top y$ and for $k=0,\ldots, K-1$ we have the recursive update
\begin{equation}
\x{k+1} =  \x{k} + \eta A^\top(y-A\x{k}) - \eta R_\theta(\x{k}).
\label{eq:ziter}
\end{equation}
Training attempts to minimize the cost function $\sum_{i=1}^n \|\xhat^{(K)}(y_i;\theta) -x_i^*\|_2^2$ with respect to the network weights $\theta$. This form of training is often called ``end-to-end''; that is, we do not train the network $R_\theta$ that replaces $\nabla r$ in isolation, but rather on the quality of the resulting estimate $\xhat^{(K)}$, which depends on the forward model $A$. Above we assume that all instances of $R_\theta$ have identical weights $\theta$, although other works explore variants where the $R_\theta$ has iteration dependent weights \cite{adler2018learned}. 

The number of iterations in deep unrolling methods is kept small for two reasons. First, at deployment, these systems are optimized to compute image estimates quickly -- a desirable property we wish to retain in developing new methods. Second, it is challenging to train deep unrolled networks for many iterations due to memory limitations of GPUs because the memory required to calculate the backpropagation updates scales linearly with the number of unrolled iterations. 

As one potential workaround, suppose we train a deep unrolled method for small number of iterations $K$ (e.g., $K=5$), then extract the learned regularizer gradient $R_\theta$ and at inference time run the iterative scheme \eqref{eq:ziter} until convergence (\ie for more iterations $K$ than used in training). Our numerical results highlight how poorly this approach performs in practice (\cref{sec:results}). Choosing a sufficiently large number of iterations $K$ (and hence the computational budget for inference) at training time is essential. As we illustrate in Fig. \ref{fig:e2eiters}, one cannot deviate from the choice of $K$ used in training and expect good performance.

\subsection{Review of Deep Equilibrium Models}
\label{sec:DEM}
In \cite{bai2019deep}, the authors propose a method for training arbitrarily-deep networks given by the repeated application of a single layer. More precisely, consider an $L$-layer network with input $y$ and weights $\theta$. Letting $\x{k}$ denote the output of the $k^{\rm th}$ hidden layer, we may write 
$$\x{k+1} = f^{(k)}_\theta(\x{k};y)\; \text{ for }\; k = 0,\ldots,L-1$$
where $k$ is the layer index and $f_\theta^{(k)}$ is a nonlinear transformation such as inner products followed by the application of a nonlinear activation function. Recent prior work explored forcing this transformation at each layer to be the same (i.e. {\em weight tying}), so that $f^{(k)}_\theta = f_\theta$ for all $k$ and showed that such networks still yield competitive performance \cite{dabre2019recurrent,bai2018trellis}. Under weight tying, we have the recursion 
\begin{equation}
    \x{k+1} = f_\theta(\x{k}; y) \label{eq:dem}.
\end{equation}
The limit of $\x{K}$ as $K \rightarrow \infty$, provided it exists, is a fixed point of the operator $f_\theta(\cdot,y)$. In \cite{bai2019deep} the authors show that the network weights $\theta$ can be learned with constant memory using implicit differentiation, bypassing computation and numerical stability issues associated with related techniques on large-scale problems \cite{chen2018neural,haber2017stable}. This past work focused on sequence models and time-series tasks, assuming that each $f_\theta$ was a single layer of a neural network, and did not explore the image reconstruction task that is the focus of this paper. 
Following the posting of a preprint of this paper, \cite{heaton2021feasibility}
propose  ``fixed-point networks'' using a strategy similar to ours, independently verifying the potential of this framework in image reconstruction.

\subsection{Plug-and-Play and Regularization by Denoising Methods} \label{sec:pnp}

Initiated by \cite{venkatakrishnan2013plug}, a collection of methods based on the {\em plug-and-play ({\sc PnP})} framework have been proposed, allowing denoising algorithms to be used as priors for model-based image reconstruction. The  starting point of {\sc PnP} is to write reconstructed image as the minimizer of a cost function given by a sum of a data-fit term and a regularizer as in \eqref{eq:opt}. Applying alternating directions method of multipliers (ADMM, \cite{boyd2011distributed,chan2016plug}) to this minimization problem gives a collection of update equations, one of which has the form
$$\argmin_x \frac{1}{2\sigma}\|z-x\|_2^2 + r(x),$$
where $r(x)$ is the regularizer and $\sigma>0$ is a parameter; this update can be considered as a ``denoising'' of the image $z$. {\sc PnP} methods replace this explicit optimization step with a ``plugged-in'' denoising method. Notably, some state-of-the-art denoisers (e.g., BM3D \cite{dabov2007image} and U-nets \cite{ronneberger2015u}) do not have an explicit $r$ associated with them, but nevertheless empirically work well within the {\sc PnP} framework. A related framework called {\em Regularization by Denoising (RED)} \cite{romano2017little} is based on a similar philosophy as {\sc PnP}, but instead
considers an explicit regularizer of the form
$$r(x) = x^\top(x-\rho(x)),$$
where $\rho(x)$ corresponds to an image denoising function.

Recent {\sc PnP} and RED efforts focuses on using training data to {\em learn} denoisers \cite{meinhardt2017learning,ryu2019plug,zhang2017learning,tirer2019super,liu2020rare}. In contrast to the unrolling methods described in \cref{sec:unrolling}, these methods are {\em not} trained end-to-end; rather, the denoising module is trained independent of the inverse problem at hand (\ie independent of the forward model $A$). As described by \cite{ongie2020deep}, decoupling the training of the learned component from $A$ results in a reconstruction system that is flexible and does not need to be re-trained for each new $A$, but can require substantially more training samples to achieve the reconstruction accuracy of a method trained end-to-end for a specific $A$.

\section{Proposed Approach}
Our approach is to design an iteration map $f_\theta(\cdot\, ; y)$ so that a fixed-point $\xfp$ satisfying
\begin{equation}
    \xfp = f_\theta(\xfp;y)
\end{equation}
is a good estimate of the image $\xstar$ given its measurements $y$. 

Here we describe choices of $f_\theta$ (and hence of the implicit infinite-depth neural network architecture) that explicitly account for the forward model $A$ and generally for the inverse problem at hand. 
Specifically, we propose choosing $f_\theta$ based on different optimization algorithms applied to regularized least squares problem \eqref{eq:opt}. This approach is similar to a DU approach (see \ref{sec:DU}), but where the number of iterations is effectively infinite -- a paradigm that has been beyond the reach of all previous deep unrolling architectures for solving inverse problems in imaging.
Below we consider three specific choices of $f_\theta$, but we note that many other options are possible. 

\begin{figure}[ht!]
    \centering
\includegraphics[width=.5\linewidth]{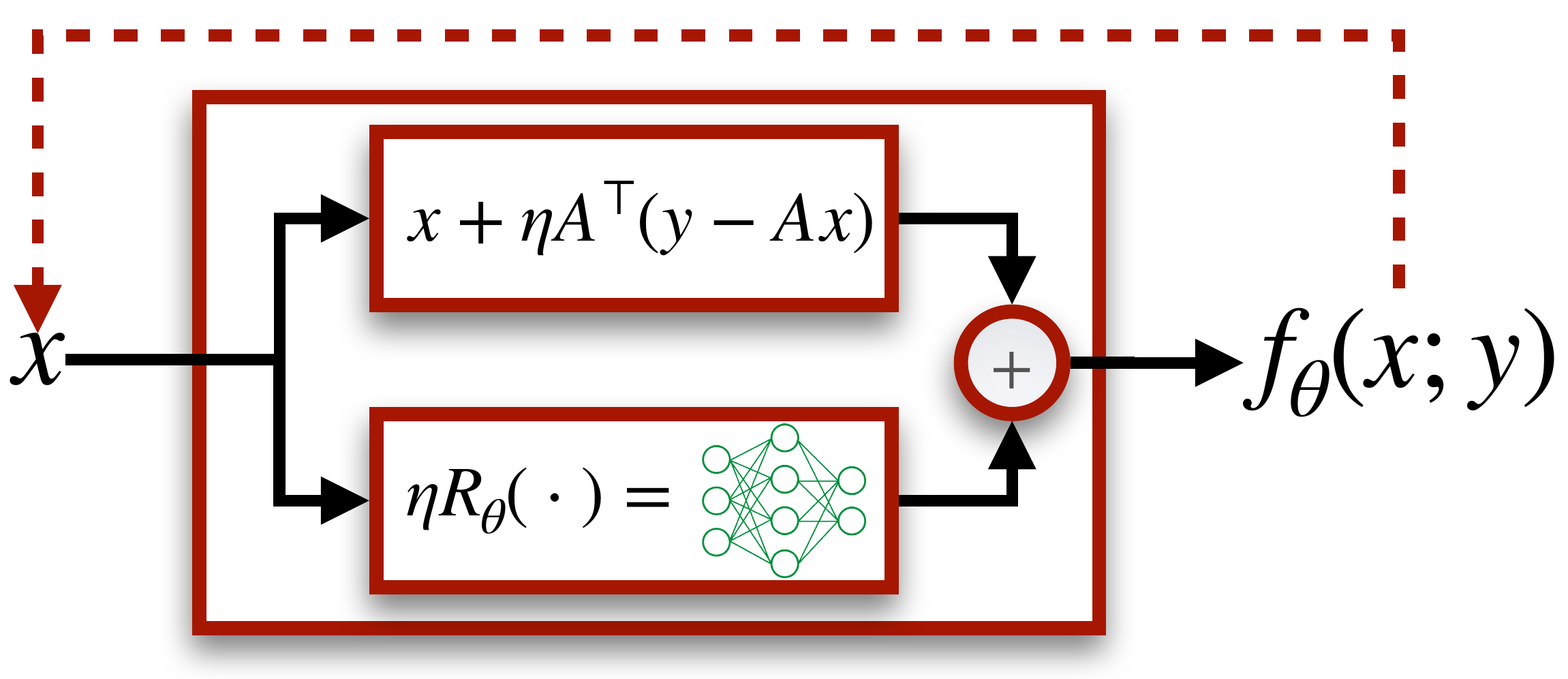}
\vspace{-1em}
    \caption{Deep Equilibrium Gradient Descent ({\sc DE-Grad)}}
    \label{fig:degrad}
\end{figure}

\subsection{Deep Equilibrium Gradient Descent ({\sc DE-Grad})}
Connecting the unrolled gradient descent iterations in \eqref{eq:ziter} 
with the deep equilibrium model in  \eqref{eq:dem}, we let
\begin{equation}
    f_\theta(x;y) = x + \eta A^\top(y-Ax) - \eta R_\theta(x).
    \label{eq:est}
\end{equation}
Recall that in this setting $R_\theta$ is a trainable network that replaces the gradient of the regularizer. 
See Figure \ref{fig:degrad} for a block diagram illustrating this choice of $f_\theta$. 


\begin{figure}[h!]
    \centering
\includegraphics[width=.7\columnwidth]{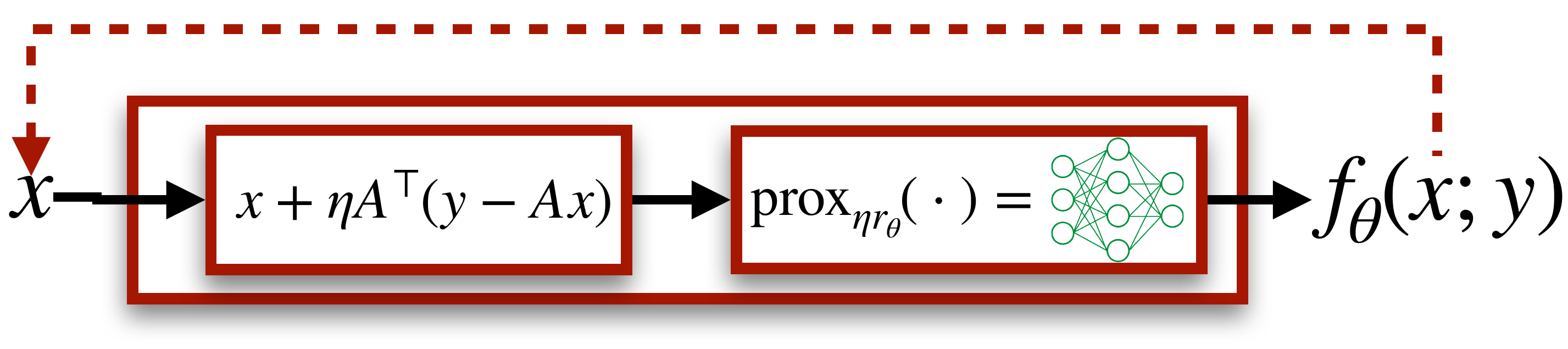}
\vspace{-1em}
    \caption{Deep Equilibrium Proximal Gradient Descent ({\sc DE-Prox})}
    \label{fig:deprox}
\end{figure}

\subsection{Deep Equilibrium Proximal Gradient Descent ({\sc DE-Prox})}

Proximal gradient methods \cite{parikh2014proximal} use a {\em proximal operator} associated with a function $h$:
\begin{equation}
    \prox_h(x) = \argmin_u \frac{1}{2}\|u-x\|_2^2 + h(u).
\end{equation}
Specifically, the proximal gradient descent algorithm applied to the optimization problem in \eqref{eq:opt}
yields the iterates
 $$\x{k+1} = \prox_{\eta r}(\x{k} + \eta A^\top (y-A\x{k})),$$ where $\eta > 0$ is a step size.
Similar to the deep unrolling approach of \cite{mardani2018neural}, we consider replacing $\prox_{\eta r}$ with a trainable network $R_\theta:\R^n\rightarrow\R^n$, which gives the iteration map
\begin{equation}\label{eq:deproxitermap}
    f_\theta(x;y) = R_\theta(x + \eta A^\top (y-Ax)).
\end{equation}
See Figure \ref{fig:deprox} for a block diagram illustrating this choice of $f_\theta$.

\begin{figure}[h!]
    \centering
\includegraphics[width=0.7\linewidth]{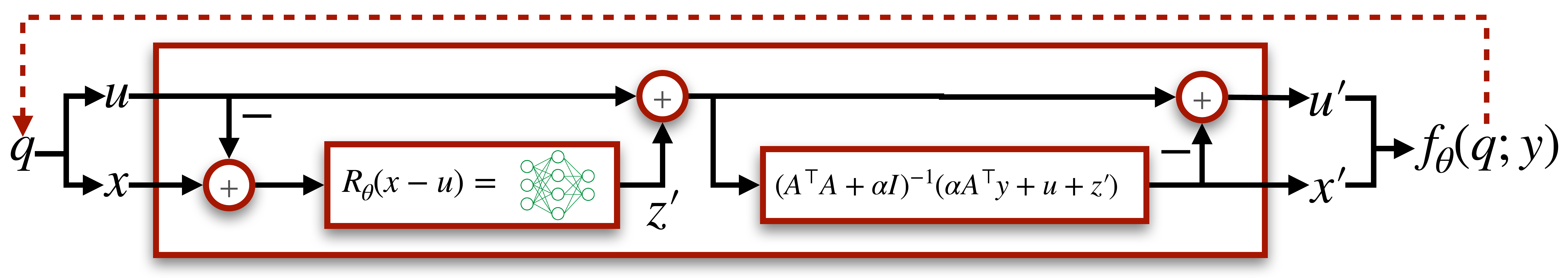}
\vspace{-1.5em}
    \caption{Deep Equilibrium Alternating Direction Method of Multipliers ({\sc DE-ADMM})}
    \label{fig:deadmm}
\end{figure}

\subsection{Deep Equilibrium Alternating Directions Method of Multipliers ({\sc DE-ADMM})}

The Alternating Directions Method of Multipliers (ADMM, \cite{boyd2011distributed}) is an efficient first-order algorithm for large-scale constrained optimization problems. ADMM can be applied to the uncontrained optimization problem \eqref{eq:opt} by rewriting it as the equivalent constrained problem
$$\min_{x,z} \frac{1}{2}\|y-Ax\|_2^2 + r(z) \text{ subject to } z = x.$$
The augmented Lagrangian (in its ``scaled form'' -- see \cite{boyd2011distributed}) associated with this problem is given by
$$
L_\alpha(x,z,u) := \frac{1}{2}\|y-Ax\|_2^2 + r(z) + \frac{1}{2\alpha}\|z-x + u\|_2^2
$$
where $u$ is an additional auxiliary variable and $\alpha>0$ is a user-defined parameter.
The ADMM iterates are then 
\begin{align}
    \begin{split}
z^{(k+1)}=&\argmin_z L_\alpha(\x{k},z,u^{(k)})\\
x^{(k+1)}=&\argmin_x L_\alpha(x,z^{(k+1)},u^{(k)})\\
u^{(k+1)}=&u^{(k)}+z^{(k+1)}-x^{(k+1)},
    \end{split}
    \label{eq:admm0}
\end{align}
Here the $z$- and $x$-updates simplify as
\begin{align*}
z^{(k+1)}=& \prox_{\alpha r}(x^{(k)}-u^{(k)})\\
x^{(k+1)}=&(I + \alpha A^\T A)^{-1}(\alpha A^\top y + z^{(k+1)} + u^{(k)}).
\end{align*}
As in the {\sc DE-Prox} approach, $\prox_{\alpha r}(\cdot)$ can be replaced with a learned network, denoted $R_\theta$. Making this replacement, and substituting $z^{(k+1)}$ directly into the expressions for $x^{(k+1)}$ and $u^{(k+1)}$ gives:
\begin{align}
\nonumber x^{(k+1)}=&(I + \alpha A^\T A)^{-1}(\alpha A^\top y + R_\theta(z^{(k)}-u^{(k)}) + u^{(k)})\\
u^{(k+1)}=&u^{(k)}+R_\theta(z^{(k)}-u^{(k)})-x^{(k+1)}.
    \label{eq:admm}
\end{align}
Note that the updates for $x^{(k+1)}$ and $u^{(k+1)}$  depend only on the previous iterates $x^{(k)}$ and $u^{(k)}$. Therefore, the above updates can be interpreted as fixed-point iterations on the joint variable $q = (x,u)$, where the iteration map $f_\theta(q;y)$ is implicitly defined as the map that satisfies 
\begin{equation}\label{eq:admmitermap}
q^{(k+1)} = f_\theta(q^{(k)},y)~\text{with}~ q^{(k)} := (z^{(k)}, u^{(k)}).
\end{equation}
Here we take the estimated image to be $\x^{\infty}$, where  $q^{(\infty)} = (x^{(\infty)}, u^{(\infty)})$ is a fixed-point of $f_\theta(\cdot;y)$. See Figure \ref{fig:deadmm} for a block diagram illustrating this choice of $f_\theta$.

\section{Calculating forward passes and gradient updates} \label{sec:calculation}

Given a choice of iteration map $f_\theta(\cdot; y)$ defining a DEM, we confronted the following obstacles. (1) Forward calculation: given an observation $y$ and network weights $\theta$, we need to be able to compute a fixed point of $f_\theta(\cdot ;y)$ {\em efficiently}. (2) Training: given a collection of training samples $\{\xstar_i\}_{i=1}^n$, we need to find the optimal network weights $\theta$. 

\subsection{Calculating Fixed-Points}

Both training and inference in a DEM require calculating a fixed point of the iteration map $f_\theta(\cdot;y)$ given some initial point $y$.
The most straightforward approach is to use fixed-point iterations given in \eqref{eq:dem}. Convergence of this scheme for specific $f_\theta$ designs is discussed in Section \ref{sec:convergence}. 

However, fixed-point iterations may not converge quickly. By viewing unrolled deep networks as fixed-point iterations, we inherit the ability to accelerate inference with standard fixed-point accelerators. To our knowledge, this work is the first time iterative inversion methods incorporating deep networks have been accelerated using fixed-point accelerators.

\paragraph{Anderson Acceleration}
Anderson acceleration
\cite{walker2011anderson}\footnote{Anderson acceleration for Deep Equilibrium models was introduced in a NeurIPS tutorial by \cite{kolter2020tutorial}.}
utilizes past iterates to identify promising directions to move during the iterations. This takes the form of identifying a vector $\alpha^{(k)} \in \R^m$ and setting, for $\beta > 0$
\begin{equation*}
    x^{(k+1)} = (1-\beta) \sum_{i=0}^{m-1} \alpha^{(k)}_i x^{(k-i)} + \beta \sum_{i=0}^{m-1} \alpha^{(k)}_i f_\theta (x^{(k-i)}; y).
\end{equation*}
The vector $\alpha^{(k)}$ is the solution to the optimization problem:
\begin{equation}
\argmin_\alpha ||G\alpha||_2^2, \ \ \ \textrm{s.t.}  \ \ \ {\bm 1}^\top \alpha = 1 \label{eq:alpha}
\end{equation}
where $G$ is a matrix whose $i^{\rm th}$ column is the (vectorized) residual $f_\theta(x^{(k-i)}; y) - x^{(k-i)}$, with $i=0,...,m-1$. The optimization problem in  \eqref{eq:alpha} admits a least-squares solution, adding negligible computational overhead when $m$ is small (\eg $m=5$).

In Section \ref{sec:results} we compare the performance and time characteristics of solving using Anderson acceleration with standard fixed-point iterations, as well as a technique which uses Broyden's method (a quasi-Newton algorithm) to find fixed points, as proposed in \cite{bai2019deep}.

An important practical consideration is that accelerating fixed-point iterations arising from optimization algorithms with auxiliary variables (like ADMM) is non-trivial. Our implementation of {\sc DE-ADMM} accelerates ADMM using the results of \cite{zhang2019accelerating}. However, in general acceleration is not required to learn to solve inverse problems, and for other algorithms or settings standard fixed-point iterations may be attractive for their simplicity of implementation. 

\subsection{Gradient Calculation} \label{sec:gradientcalc}

In this section, we provide a brief overview of the training procedure used to train all networks in Section \ref{sec:results}. We use stochastic gradient descent to find network parameters $\theta$ that (locally) minimize a cost function of the form $\frac{1}{n}\sum_{i=1}^n \ell(\xfp(y_i;\theta),x_i^*)$
where $\ell(\cdot,\cdot)$ is a given loss function, $x_i^*$ is the $i$th training image with paired measurements $y_i$, and $\xfp(y_i;\theta)$ denotes the reconstructed image given as the fixed-point of $f_\theta(\cdot\,; y_i)$. For our image reconstruction experiments, we use the mean-squared error (MSE) loss:
\begin{equation}
    \ell(x, \xstar) = \frac{1}{2} ||x - \xstar||_2^2.
\end{equation}

To simplify the calculations below, we consider gradients of the cost function with respect to a single training measurement/image pair, which we denote $(y,x^*)$.
Following \cite{bai2019deep}, we leverage the fact that $\xfp := \xfp(y;\theta)$ is a fixed-point of $f_\theta(\cdot;y)$ to find the gradient of the loss with respect to the network parameters $\theta$ without backpropagating through an arbitrarily-large number of fixed-point iterations. We summarize this approach below.

First, abbreviating $\ell(\xfp,\xstar)$ by $\ell$, then by the chain rule the gradient of $\ell$ with respect to the network parameters is given by
\begin{equation} \label{eq:genericgrad}
   \frac{\partial \ell}{\partial \theta} =  \frac{\partial \xfp}{\partial \theta}^\top \frac{\partial \ell}{\partial \xfp}.
\end{equation}
where $\frac{\partial \xfp}{\partial \theta}$ is the Jacobian of $\xfp$ with respect to $\theta$, and $\frac{\partial \ell}{\partial \xfp}$ is the gradient of $\ell$ with respect to its first argument evaluated at $\xfp$. Since we assume $\ell$ is the MSE loss, the gradient $\frac{\partial \ell}{\partial \xfp }$ is simply the residual between $\xstar$ and the equilibrium point: $\frac{\partial \ell}{\partial \xfp} = \xfp-\xstar$. 

Now, in order to compute the Jacobian $\frac{\partial \xfp}{\partial \theta}$ we start with the fixed point equation: $\xfp = f_\theta(\xfp; y)$. Differentiating both sides of this equation, and solving for $\frac{\partial \xfp}{\partial \theta}$ gives
\begin{align}
    \frac{\partial \xfp}{\partial \theta} &= \left(I - \left.\frac{\partial f_\theta(x; y)}{\partial x}\right|_{x=\xfp}\right)^{-1} \frac{\partial f_\theta(\xfp; y)}{\partial \theta}.
\end{align}
Plugging this expression into \eqref{eq:genericgrad} gives
\begin{equation} \label{eq:fullgrad}
    \frac{\partial \ell}{\partial \theta} = \frac{\partial f_\theta(\xfp; y)}{\partial \theta}^\top \left(I - \frac{\partial f_\theta(x; y)}{\partial x} \bigg\rvert_{x=\xfp} \right)^{-\top}\!\!\!\! (\xfp - \xstar) \nonumber
\end{equation}
This converts the memory-intensive task of backpropagating through many iterations of $f_\theta(\cdot\,; y)$ to the problem of calculating an inverse Jacobian-vector product. To approximate the inverse Jacobian-vector product, first we define the vector $\beta^{(\infty)}$ by
\begin{equation*}
    \beta^{(\infty)} = \left(I - \frac{\partial f_\theta(x; y)}{\partial x} \bigg\rvert_{x=\xfp} \right)^{-\top} (\xfp - \xstar).
\end{equation*}
Following \cite{kolter2020tutorial}, we note that $\beta = \beta^{(\infty)}$ is a fixed point of the equation
\begin{equation} \label{eq:backpropfixedpoint}
    \beta = \left(\frac{\partial f_\theta(x; y)}{\partial x} \bigg\rvert_{x=\xfp}\right)^\top \beta + (\xfp - \xstar),
\end{equation}
and the same machinery used to calculate the fixed point $\xfp$ may be used to calculate $\beta^{(\infty)}$. For analysis purposes, we note that the limit of fixed-point iterations for solving \eqref{eq:backpropfixedpoint} with initial iterate  $\beta^{(0)} = {0}$ is equivalent to the Neumann series:
\begin{equation} \label{eq:backpropneumann}
    \beta^{(\infty)} = \sum_{n=0}^\infty\left[ \left(\frac{\partial f_\theta(x; y)}{\partial x} \bigg\rvert_{x=\xfp}\right)^{\top}\right]^n (\xfp - \xstar).
\end{equation}
Convergence of the above Neumann series is discussed in Section \ref{sec:convergence}.

Conventional autodifferentiation tools permit quickly computing the vector-Jacobian products in \eqref{eq:backpropfixedpoint} and \eqref{eq:backpropneumann}. Once an accurate approximation to  $\beta^{(\infty)}$ is calculated, the gradient in \eqref{eq:genericgrad} is given by
\begin{equation}
    \frac{\partial \ell}{\partial \theta} = \frac{\partial f(\xfp; y)}{\partial \theta}^\top \beta^{(\infty)}.
\end{equation}
The gradient calculation process is summarized in the following steps, assuming a fixed point $\xfp$ of $f_\theta(\cdot\,; y)$ is known:
\begin{enumerate}[topsep=-2ex,itemsep=-1ex,partopsep=1ex,parsep=1ex,leftmargin=4ex]
    \item Compute the residual $r = x^\infty-x^*$.
    \item Compute an approximate fixed-point $\beta^{(\infty)}$ of the equation $\beta = \left(\frac{\partial f_\theta(x; y)}{\partial x} \big\rvert_{x=\xfp}\right)^\top \beta + r$.
    \item Compute $\frac{\partial \ell}{\partial \theta} = \frac{\partial f_\theta(\xfp; y)}{\partial \theta}^\top \beta^{(\infty)}$.
\end{enumerate}

\section{Convergence Theory} \label{sec:convergence}
Here we study convergence of the proposed deep equilibrium models to a fixed-point at inference time, \ie given the iteration map $f_\theta(\cdot\, ;y):\R^d\rightarrow\R^d$ we give conditions that guarantee the convergence of the iterates $x^{(k+1)} = f_\theta(x^{(k)};y)$ to a fixed-point $\xfp$ as $k\rightarrow\infty$.

Classical fixed-point theory ensures that the iterates converge to a unique fixed-point if the iteration map $f_\theta(\cdot; y)$ is \emph{contractive}, \ie if there exists a constant $0 \leq c < 1$ such that $\|f_\theta(x;y)-f_\theta(x';y)\| \leq c \|x-x'\|$. 
Below we give conditions on the regularization network $R_\theta:\R^d\rightarrow\R^d$ (replacing the gradient or proximal mapping of a regularizer) used in the {\sc DE-Grad}, {\sc DE-Prox} and {\sc DE-ADMM} models that that ensure the resulting iteration map is contractive and thus the fixed-point iterations for these models converge.

In particular, following \cite{ryu2019plug}, we assume that the regularization network $R_\theta$  satisfies the following condition: there exists an $\epsilon > 0$ such that for all $x,x'\in \R^d$ we have
\begin{equation}
    \|(R_\theta-I)(x)-(R_\theta-I)(x')\| \leq \epsilon \|x-x'\|
    \label{eq:lip}
\end{equation}
where $(R_\theta-I)(x):= R_\theta(x) - x$. In other words, we assume the map $R_\theta-I$ is $\epsilon$-Lipschitz. 

If we interpret $R_\theta$ as a denoising or de-artifacting network, then $R_\theta-I$ is the map that outputs the noise or artifacts present in a degraded image. In practice, often $R_\theta$ is implemented with a residual ``skip-connection'', such that $R_\theta = I + N_\theta$, where $N_\theta$ is, e.g., a deep U-net. Therefore, in this case, \eqref{eq:lip} is equivalent to assuming the trained network $N_\theta$ is $\epsilon$-Lipschitz.

First, we have the following convergence result for {\sc DE-Grad}:

\begin{thm}[Convergence of {\sc DE-Grad}] Assume that $R_\theta - I$ is $\epsilon$-Lipschitz \eqref{eq:lip}, and let $L = \lambda_{\max}(A^\T A)$ and $\mu = \lambda_{\min}(A^\T A)$, where $\lambda_{\max}(\cdot)$ and $\lambda_{\min}(\cdot)$ denote the maximum and minimum eigenvalue, respectively. If the step-size parameter $\eta > 0$ is such that $\eta < 1/(L+1)$, then the {\sc DE-Grad} iteration map $f_\theta(\cdot;y)$ defined in \eqref{eq:est} satisfies
\begin{equation*}
    \|f_\theta(x;y)-f_\theta(x';y)\| \leq \underbrace{(1-\eta(1+\mu) + \eta \epsilon)}_{=:\gamma} \|x-x'\|
    \vspace{-1em}
\end{equation*}
for all $x,x'\in \R^d$. The coefficient $\gamma$ is less than $1$ if $\epsilon < 1 + \mu$, in which case the the iterates of {\sc DE-Grad} converge.
\end{thm}
\begin{proof}
Let $f_\theta(x;y)$ be the iteration map for {\sc DE-Grad}. The Jacobian of $f_\theta(x;y)$ with respect to $x\in\R^d$, denoted by $\partial_x f_\theta(x;y)$, is given by
\[
\partial_x f_\theta(x;y) = (I-\eta A^\T A)-\eta \partial_x R_\theta(x)  \in \R^{d\times d}
\]
where $\partial_x R_\theta(x) \in \R^{d\times d}$ is the Jacobian of $R_\theta:\R^d\rightarrow\R^d$ with respect to $x \in \R^d$. To prove $f_\theta(\cdot\,;y)$ is contractive it suffices to show $\|\partial_x f_\theta(x;y)\| < 1$ for all $x\in\R^d$ where $\|\cdot\|$ denotes the spectral norm.
Towards this end, we have
\begin{align}
    \|\partial_x f_\theta(x;y)\| & = \|(I-\eta A^\T A)-\eta \partial_x R_\theta(x)\|\nonumber\\
    & = \|\eta I +(1-\eta) I-\eta A^\T A-\eta \partial_x R_\theta(x)\|\nonumber\\
    & = \|(1-\eta) I-\eta A^\T A-\eta (\partial_x R_\theta(x)-I)\|\nonumber\\
    & \leq \|(1-\eta) I-\eta A^\T A\| + \eta\|\partial_x R_\theta(x)-I\|\nonumber\\
    & \leq \max_i |(1-\eta)-\eta\lambda_i| + \eta\epsilon\label{eq:finalineq}
\end{align}
where $\lambda_i$ denotes the $i$th eigenvalue of $A^\T A$, and in the final inequality \eqref{eq:finalineq} we used our assumption that the map $(R_\theta-I)(x):= R_\theta(x)-x$ is $\epsilon$-Lipschitz, and therefore the spectral norm of its Jacobian  $\partial_xR_\theta(x)-I$ is bounded by $\epsilon$.

Finally, by our assumption $\eta < \frac{1}{1+L}$ where $L := \max_i \lambda_1$, we have $\eta < \frac{1}{1+\lambda_i}$ for all $i$, which implies
$(1-\eta)-\eta \lambda_i > 0$ for all $i$. Therefore, the maximum in \eqref{eq:finalineq} is obtained at $\mu := \min_i \lambda_i$, which gives
\[
\|\partial_x f_\theta(x;y)\|\leq 1-\eta(1+\mu) + \eta\epsilon.
\]
This shows $f_\theta$ is $\gamma$-Lipschitz with $\gamma = 1-\eta(1+\mu) + \eta\epsilon$, proving the claim.\end{proof}

Convergence of {\sc PnP} approaches {\sc PnP-Prox} and {\sc PnP-ADMM} is studied in \cite{ryu2019plug}. At inference time, the proposed {\sc DE-Prox} and {\sc DE-ADMM} methods are equivalent to the corresponding {\sc PnP} method but with a retrained denoising network $R_\theta$. Therefore, the convergence results in \cite{ryu2019plug} apply directly to {\sc DE-Prox} and {\sc DE-ADMM}. To keep the paper self-contained, we restate these results below, specialized to the case of the quadratic data-fidelity term assumed in \eqref{eq:opt}.

\begin{thm}[Convergence of {\sc DE-Prox}] Assume that $R_\theta - I$ is $\epsilon$-Lipschitz \eqref{eq:lip}, and let $L = \lambda_{\max}(A^\T A)$ and $\mu = \lambda_{\min}(A^\T A) > 0$, where $\lambda_{\max}(\cdot)$ and $\lambda_{\min}(\cdot)$ denote the maximum and minimum eigenvalue, respectively. Then the {\sc DE-Prox} iteraion map $f_\theta(\cdot,y)$ defined in \eqref{eq:deproxitermap} is contractive if the step-size parameter $\eta$ satisfies
\begin{equation*}
\frac{1}{\mu(1+1 / \varepsilon)}<\eta<\frac{2}{L}-\frac{1}{L(1+1 / \varepsilon)}.
\end{equation*}
Such an $\eta$ exists if $\varepsilon<2 \mu /(L-\mu)$.
\end{thm}
See Theorem 1 of \cite{ryu2019plug}. 

\begin{thm}[Convergence of {\sc DE-ADMM}]  Assume that $R_\theta - I$ is $\epsilon$-Lipschitz \eqref{eq:lip}, and let $L = \lambda_{\max}(A^\T A)$ and $\mu = \lambda_{\min}(A^\T A) > 0$, where $\lambda_{\max}(\cdot)$ and $\lambda_{\min}(\cdot)$ denote the maximum and minimum eigenvalue, respectively. Then the iteration map $f_\theta(\cdot;y)$ for {\sc DE-ADMM} defined in \eqref{eq:admmitermap} is contractive if the ADMM step-size parameter $\alpha$ parameter satisfies
\begin{equation*}
    \frac{\varepsilon}{(1+\varepsilon-2\varepsilon^2)\mu} < \alpha.
\end{equation*}
\end{thm}
See Corollary 1 of \cite{ryu2019plug}. 
 
Unlike the convergence result for {\sc DE-Grad} given in Theorem 1, the convergence results for {\sc DE-Prox} and {\sc DE-ADMM} in Theorem 2 and Theorem 3 make the assumption that $\lambda_{min}(A^\top A)>0$, i.e., $A$ has a trivial nullspace. This is condition is satisfied for certain inverse problems, such as denoising or deblurring, but violated in many others, including compressed sensing and undersampled MRI. However, in practice we observe that the iterates of {\sc DE-Prox} and {\sc DE-ADMM} still appear to converge even in situations where $A$ has a nontrivial nullspace, indicating this assumption may be stronger than necessary.

Finally, an important practical concern when training deep equilibrium models is whether the fixed-point iterates  used to compute gradients (as detailed in Section \ref{sec:gradientcalc}) will converge. Specifically, the gradient of the loss at the training pair $(y,x^*)$ involves computing the truncated Neumann series in \eqref{eq:backpropneumann}. This series converges if the Jacobian $\partial_x f_\theta(x;y)$ has spectral norm strictly less than $1$ when evaluated at any $x \in \R^d$, which is true if and only if the iteration map $f_\theta(\cdot\,;y)$ is contractive. Therefore, the same conditions in Theorems 1-3 that ensure the iteration map is contractive also ensure that the Neumann series in \eqref{eq:backpropneumann} used to compute gradients converges.

\section{Experimental Results}

\subsection{Comparison Methods and Inverse Problems}

\begin{table*}[t]
\caption{Mean PSNR and SSIM over test set; the highest PSNRs and SSIMs for each setting are in bold.}
\label{tab:reconperformance}
\begin{small}
\begin{sc}
\begin{adjustbox}{width=\linewidth,center}
\begin{tabular}{l|c|c|c|c|c|c|c|c|c|c|c}
\toprule
 & & \multicolumn{2}{p{1.0in}|}{\centering Plug-n-Play (DnCNN denoiser)} & \multicolumn{1}{p{0.7in}|}{\centering RED (DnCNN denoiser)} & \multicolumn{4}{p{2.0in}|}{\centering Deep Unrolled Methods Trained End-to-End} & \multicolumn{3}{c}{Deep Equilibrium (Ours)} \\
  & TV & Prox & ADMM & ADMM & Grad & Prox & ADMM & Neumann & Grad & Prox & ADMM \\
\hline
Deblur (1) &&&&&&&&&&&\\
 PSNR & 26.79 & 29.77 & 29.95 & 29.78 & 32.23 & 31.64 & 31.45 & 32.39 & \bf 32.43 & 31.87 & 32.30  \\ 
 SSIM & 0.86 & 0.88 & 0.89 & 0.89 & 0.93 & 0.93 & 0.93 & 0.94 & \bf 0.94 & 0.93 & 0.94    \\ \hline
 Deblur (2) &&&&&&&&&&&\\
 PSNR & 31.31 & 35.25 & 35.61 & 35.22 & 36.10 & 36.92 & 36.14 & 36.24 & \bf 37.99 & 37.84 & 37.95  \\ 
 SSIM & 0.90  & 0.96 & 0.96 & 0.96 & 0.97 & 0.97 & 0.97 & 0.97 & \bf 0.98 & 0.97 & 0.98 \\ \hline
CS (4x) &&&&&&&&&&&\\
PSNR & 26.04 & 27.79 & 27.85 & 27.80 & 29.32 & 29.65 & 29.09 & 29.59 & 31.46 & 31.51 & \bf 31.64  \\
 SSIM & 0.83 & 0.87 & 0.87 & 0.86 & 0.88 & 0.89 & 0.88 & 0.89 & 0.93 & 0.93 & \bf 0.93  \\ \hline
 MRI (8x) &&&&&&&&&&&\\
PSNR & 26.64 & 28.45 & 28.39 & 29.89 & 31.13 & 30.97 & 31.82 & 31.64 & \bf 32.01 & 31.02 & 31.38  \\
 SSIM & 0.78 & 0.85 & 0.85 & 0.88 & 0.89 & 0.88 & 0.88 & 0.89 & \bf 0.89 & 0.88 & 0.89 \\ \hline
 MRI (4x) &&&&&&&&&&&\\
PSNR & 31.22 & 31.56 & 31.92 & 32.37 & 32.44 & 32.49 & 32.56 & 32.62 & 33.41 & 33.66 & \bf 33.72  \\
 SSIM & 0.88 & 0.89 & 0.89 & 0.90 & 0.91 & 0.91 & 0.91 & 0.92 & 0.91 & 0.92 & \bf 0.92  \\
\bottomrule
\end{tabular}
\end{adjustbox}
\end{sc}
\end{small}
\end{table*}

Our numerical experiments include comparisons with a variety of models and methods. 
Total-variation Regularized Least Squares (TV) is an important baseline that does not use any training data but rather leverages geometric models of image structure \cite{rudin1992nonlinear,strong2003edge,beck2009fast}. The {\sc PnP} and {\sc RED} methods are described in \cref{sec:pnp}; we consider both the original ADMM variant of \cite{venkatakrishnan2013plug} {\sc PnP-ADMM} and a proximal gradient {\sc PnP-Prox} method as described in  \cite{ryu2019plug}. We utilize the ADMM formulation of {\sc RED}.
Deep Unrolled methods (DU) are described in \cref{sec:DU}; we consider DU using gradient descent, proximal gradient, and ADMM.
The preconditioned Neumann network \cite{gilton2019neumann} does not have simple Deep Equilibrium or Plug-and-Play analogues, and is included as an alternative deep unrolled method.

All deep unrolled methods have \emph{tied weights}, i.e. the network used at each iteration is the same. This is done to ensure that the number of parameters in the deep unrolled, deep equilibrium, and Plug-and-Play/RED methods are the same. Moreover, on modestly-sized datasets recent work has shown that tied weights result in better reconstructions \cite{aggarwal2018modl}. 

We compare the above approaches across three inverse problems: Image deblurring (Deblur), compressed sensing (CS), and accelerated MRI reconstruction (MRI). In the Deblur setting we simulate blurry images using $9\times 9$ pixel Gaussian blur kernel with variance $5$ and consider two noise levels: Deblur (1) refers to the setting of additive white Gaussian noise with variance $\sigma=0.01$ (high noise), and Deblur (2) refers the setting of additive white Gaussian noise with standard deviation $\sigma  = 0.0001$ (low noise). In the CS setting, we take linear measurements of the image by forming inner products with random Gaussian vectors whose entries i.i.d.~standard normals, and use an undersampling factor of $4\times$, \ie the corresponding forward model $A$ has $4\times$ is a random Gaussian matrix fewer rows than columns. The measurements are then corrupted with additive white Gaussian noise with standard deviation $\sigma = 0.01$.  In the MRI setting, we investigate recovery at $4\times$ and $8\times$ acceleration, where ``$\rho\times$ acceleration'' indicates an undersampling factor of $\rho$ in k-space (not to be confused with the acceleration techniques used in finding fixed points). Our MRI experiments focus on the case of (virtual) single-coil MR data acquired on a Cartesian grid in k-space, where corresponding forward model $A$ is a subsampling of rows of the discrete Fourier transform matrix. We additionally add complex white Gaussian noise to the undersampled k-space measurements with standard deviation $\sigma = 0.01$.

For the Deblur and CS problems, we utilize a subset of the Celebrity Faces with Attributes (CelebA) dataset \cite{liu2015faceattributes}, which consists of centered human faces. We train on a subset of 10000 of the training images. All images are resized to 128$\times$128. For the MRI problem, we use a random subset of size 2000 of the fastMRI single-coil knee dataset \cite{zbontar2018fastMRI} for training. We trim the ground truth MR images to a 320$\times$320 pixel region-of-interest.

\subsection{Architecture Specifics}

For our learned network, we utilize a DnCNN architecture as in \cite{ryu2019plug}. 
(We also experimented with U-Nets, but found that DnCNN yielded superior performance for both our proposed deep eq methods and the comparison methods.)
For both the CelebA and fastMRI datasets, we train six DnCNN denoisers with noise variances $\sigma^2=0.1, 0.05, 0.02, 0.01, 0.005, 0.001$ on the training split. Training follows the methodology of \cite{ryu2019plug}. Specifically, to ensure contractivity of the learned component, we add spectral normalization to all layers, ensuring that each layer has a Lipschitz constant bounded above by 1. This normalization is enforced during pretraining as well as during the Deep Equilibrium training phrase.

During training, we utilize Anderson acceleration for both the forward and backward pass fixed-point iterations. In backward passes, the number of fixed-point iterations was limited to 50 due to memory constraints, but fixed-point iterations in forward passes were run until convergence was observed (defined to be when the relative norm difference between iterations is less than $10^{-3}$). Test time results were produced using the same method as the forward pass during training, but over the test set. We compare different fixed-point calculation methods for the forward and backward passes in Section \ref{sec:accelerationstudy}.

Further details on settings, parameter choices, and data may be found in the supplementary materials and in our publicly-available code.\footnote{Available at: \url{https://github.com/dgilton/deep_equilibrium_inverse}}

\subsection{Parameter Tuning and Pretraining}
\label{sec:tune}


Each of the iterative optimization algorithms we test has its own set of hyperparameters to choose, \eg the step size $\eta$ in {\sc DE-Grad}, plus any parameters used to calculate the initial estimate $x^{(0)}$.
Tuning hyperparameters requires choosing a particular regularization network $R_\theta$ during tuning. We choose from a collection of $R_\theta$ that have been pretrained for Gaussian denoising at different noise levels. Pretraining can be done on the training dataset (\eg training on MRI images directly) or using an independent dataset (\eg the BSD500 image dataset \cite{martin2001database}). We use the former approach in our experiments.

To tune hyperparameters, we first choose parameters to optimize the performance of {\sc PnP} on a validation set via a grid search over pretrained $R_\theta$ as well as algorithm-specific hyperparameters (such as the $\eta$ step-size parameter in gradient descent approaches).
Then, we use the {\sc PnP} hyperparameter settings as initial hyperparameter values when  training Deep Equilibrium or Deep Unrolling methods. 


\subsection{Main Results}
\label{sec:results}

We present the main reconstruction accuracy comparison in Table \ref{tab:reconperformance}. Each entry for Deep Equilibrium (DE), Regularization by Denoising (RED), and Plug-and-Play ({\sc PnP}) approaches is the result of running fixed-point iterations until the relative change between iterations is less than $10^{-3}$. During training, all DE models were limited to a maximum 100 forward updates, but terminate iterations on the relative norm difference between iterations falling below $10^{-3}$. The DU models are tested at the number of iterations for which they were trained 
and all parameters for TV reconstructions (including number of TV iterations) are cross-validated to maximize PSNR. Performance as a function of iteration is shown in Figs.~\ref{fig:psnrclock},~\ref{fig:iteratesdeblur}, and~\ref{fig:iteratescs}, with example reconstructions in Fig.~\ref{fig:MRIExample}. Further example reconstructions are available for qualitative evaluation in the supplementary materials.

We observe our Deep Equilibrium-based approaches consistently outperform Deep Unrolled approaches across different choices of base algorithm (\ie {\sc Grad,Prox,ADMM}). Among choices of iterative reconstruction architectures for $f_\theta$, there does not appear to be an obvious winner, suggesting the optimal choice may be problem- or setting-dependent.

\begin{figure}[ht]
\begin{center}
\subfigure[Deblurring (2)]{\label{fig:iteratesdeblur}\includegraphics[width=0.49\linewidth]{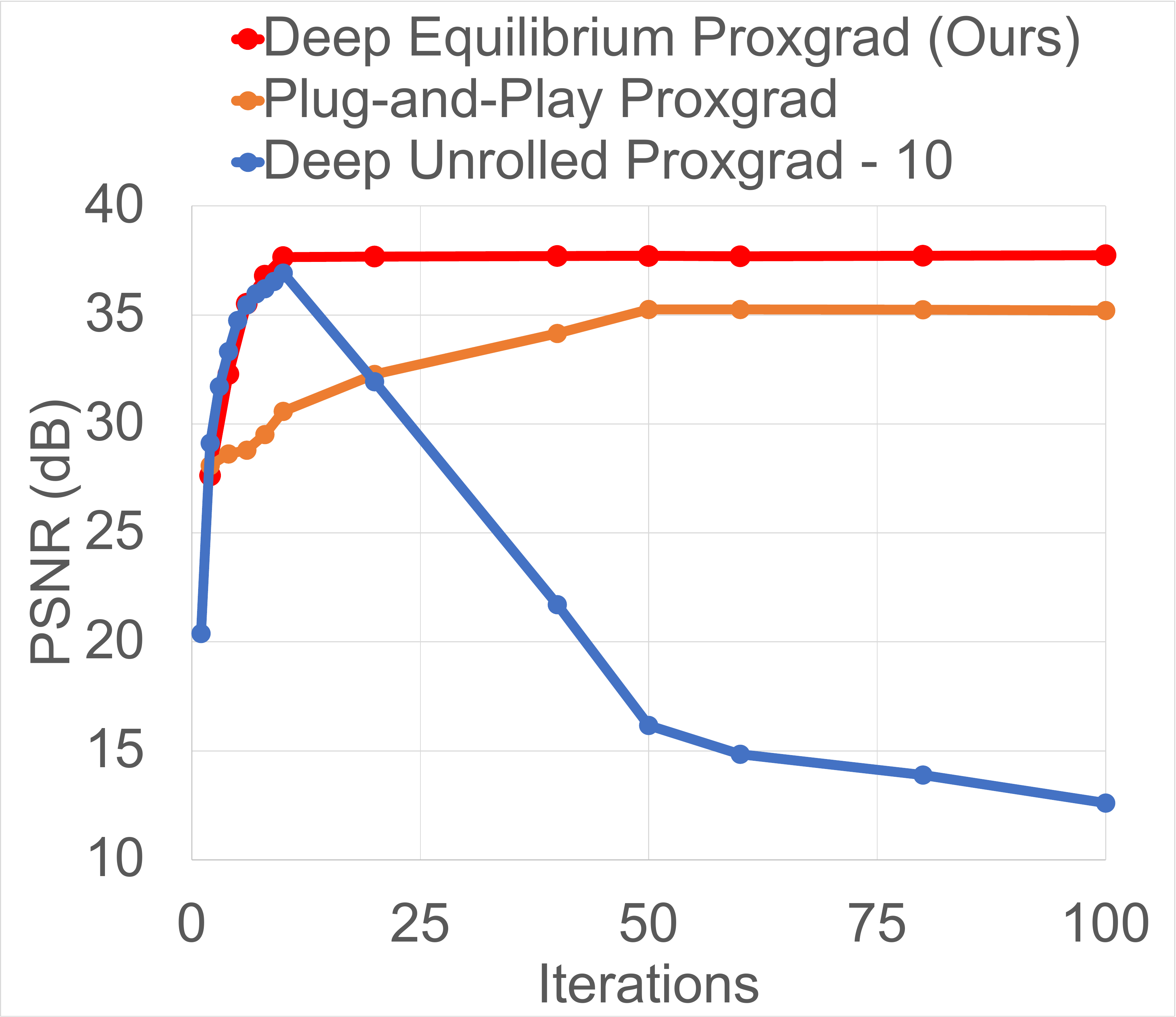}}~
\subfigure[CS $4\times$]{\label{fig:iteratescs}\includegraphics[width=0.49\linewidth]{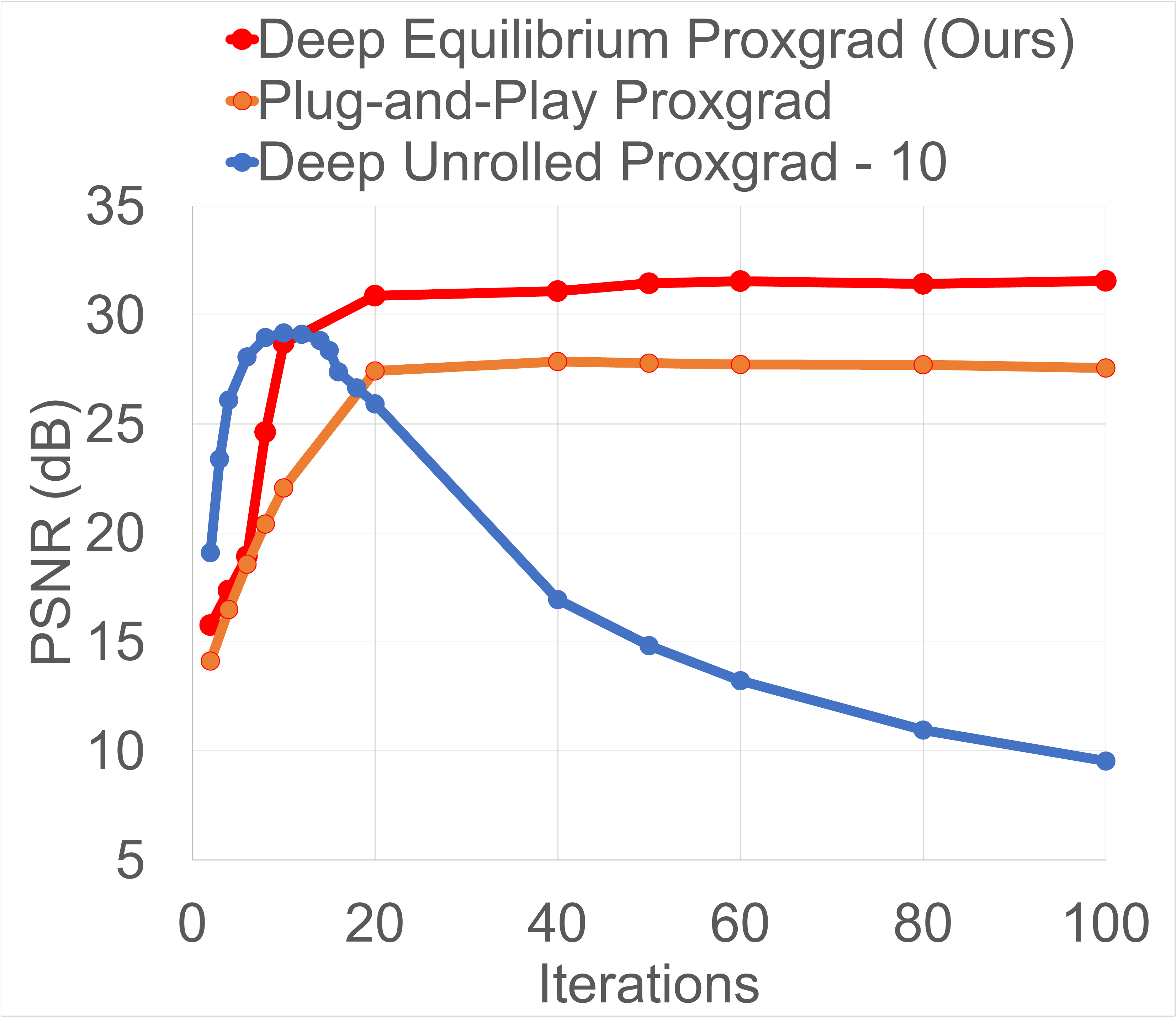}}
\vspace{-.2in}
\caption{Iterations vs. reconstruction PSNR for {\sc DE-Prox} and competing methods for (a) deblurring and (b) compressed sensing. MRI results are in Fig.~\ref{fig:psnrclock}. The deep unrolled ProxGrad was trained for 10 iterations. In all examples, deep unrolling is only effective at the number of iterations for which it is trained, whereas deep equilibrium achieves higher PSNR across a broad range of iterations, allowing a user to trade off computation time and accuracy.}
\label{fig:iterationstudy}
\end{center}
\end{figure}


\begin{figure}[ht!]
\begin{center}
\subfigure[Ground truth]{\label{fig:mri_true}\includegraphics[width=0.24\linewidth]{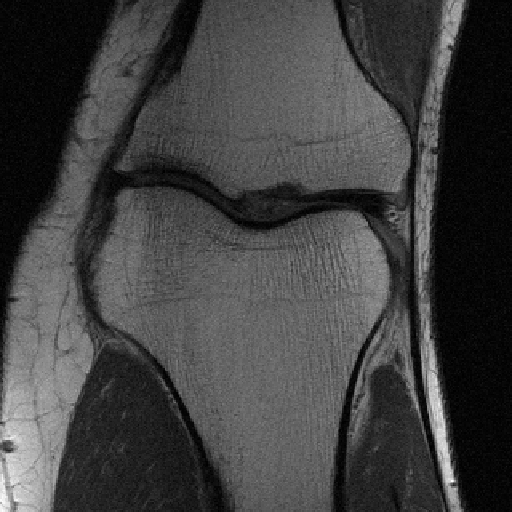}}~
\subfigure[IFFT ($A^\top y$), PSNR = 24.53 dB dB]{\label{fig:ifft}\includegraphics[width=0.24\linewidth]{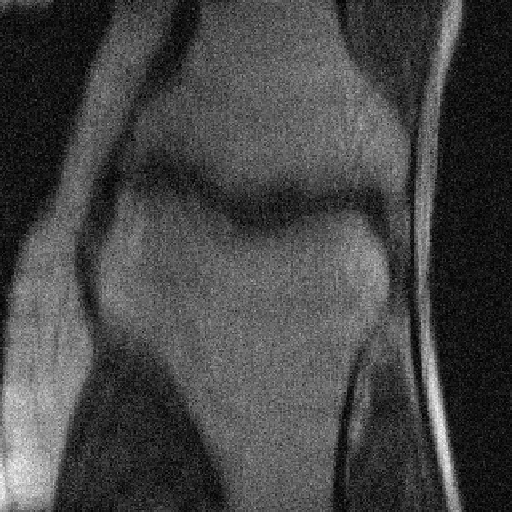}}
\subfigure[{\sc DU-Prox}, PSNR = 31.02 dB dB]{\label{fig:du_mri}\includegraphics[width=0.24\linewidth]{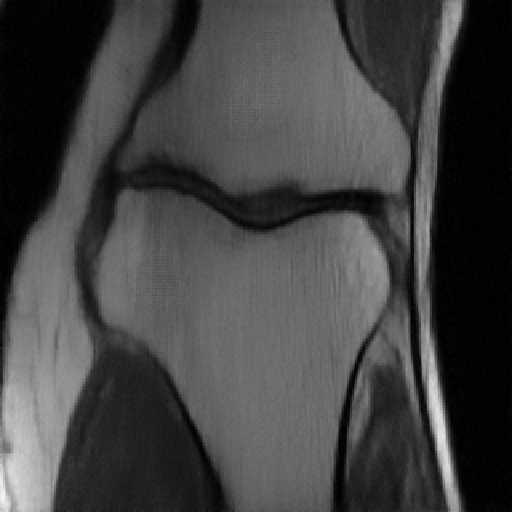}}~
\subfigure[{\sc DE-Prox}, PSNR = 32.09 dB dB]{\label{fig:de_mri}\includegraphics[width=0.24\linewidth]{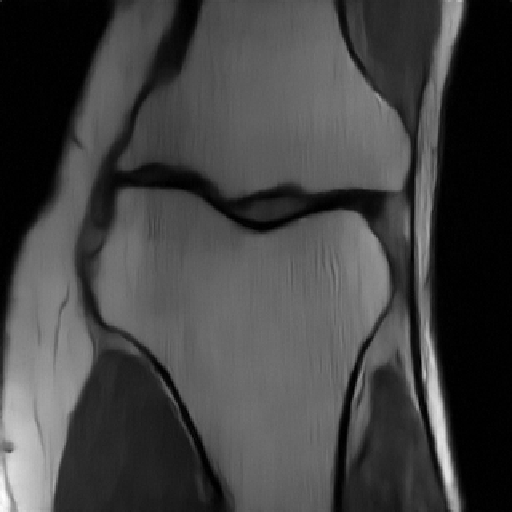}}
\vspace{-.1in}
\caption{$8\times$ accelerated MRI reconstruction example. Best viewed digitally.}
\vspace{-.1in}
\label{fig:MRIExample}
\end{center}
\end{figure}

As seen in Figs. \ref{fig:psnrclock},~\ref{fig:iteratesdeblur}, and~\ref{fig:iteratescs}, the Deep Equilibrium approach 
generally outperforms
Deep Unrolled solvers. 
Figs.~\ref{fig:psnrclock},~\ref{fig:iteratesdeblur}, and~\ref{fig:iteratescs} show that our approach
requires \emph{no more computation} than Deep Unrolled networks to achieve the same performance level and has an increasing advantage with further computation.

\subsection{Effect of Acceleration} \label{sec:accelerationstudy}

Here we demonstrate the effect of using different fixed-point solvers during both the training and inference procedures. Leveraging acceleration can decrease computational costs during both training and inference and result in better empirical performance at inference. 
Table \ref{tab:memorytimecosts} compares Anderson accelerated Deep Equilibrium approaches with Deep Unrolling, Plug and Play, Deep Equilibrium utilizing Broyden's Method (as was used in \cite{bai2019deep}), and non-accelerated Deep Equilibrium. All results were determined using PyTorch utilizing an NVidia RTX 2080 Ti GPU. As mentioned previously, non-accelerated Deep Equilibrium at inference time has identical per-iteration cost and memory requirements as Plug and Play and Deep Unrolling.

We observe that while finding the fixed-points using Broyden's method and Anderson acceleration requires more time per iterate, convergence occurs faster than a standard fixed-point iteration so the net time spent at inference is less. Since the Broyden solution was slightly worse in terms of PSNR, Anderson acceleration was used for all other experiments.

For practical matters, the memory cost of each of the compared methods may also be an important factor to consider. At train time, Deep Unrolled methods require memory scaling linearly with the number of iterations used, while Deep Equilibrium methods require only the memory necessary to compute the gradient in \eqref{eq:fullgrad}, which is what permits training at convergence. Plug and Play methods are the least memory-intensive of all to train. At inference time, each method only needs to store at most a constant number of iterations, so all methods are cheap in terms of memory to evaluate.

\begin{table}[t]
\label{tab:memorytimecosts}
\begin{adjustbox}{width=0.7\columnwidth,center}
\begin{tabular}{l|c|c|c}
\toprule
  & Total Time (s) & Time/Iteration (s) & PSNR (dB) \\ \hline
Plug $\&$ Play  & 1.24 & 0.025 & 31.56 \\ \hline
{\sc DU-Prox} & 0.25 & 0.025 & 32.49 \\ \hline
{\sc DE-Prox} & 1.22 & 0.025 & 31.86 \\ \hline
{\sc DE-Prox} (Anderson) & 1.17 & 0.046 & 33.66 \\ \hline
{\sc DE-Prox} (Broyden) & 1.85 & 0.039 & 33.61
\end{tabular}
\end{adjustbox}
\caption{Mean computation time required to reach convergence and resulting mean reconstruction PSNR in $4\times$ accelerated MRI reconstruction with complex image size $320\times320$, as computed over the test set.}
\end{table}

\subsection{Effect of Pre-Training} \label{sec:pretraining}

Here we compare the effect of initializing the learned component $R_\theta$ in our deep equilibrium models with a pretrained denoiser versus initializing with random weights. 
We use the same hyperparameter tuning scheme described in Section~\ref{sec:tune}.

We present our results on Deep Equilibrium Proximal Gradient Descent ({\sc DE-Prox}) in Figure \ref{fig:initializationstudy}. We observe an improvement in reconstruction quality when utilizing our pretraining method compared to a random initialization. We also note that pretraining enables a simple choice of algorithm-specific hyperparameters, such as the initial internal step size for {\sc DE-Prox}. 

\subsection{Noise Sensitivity}

We observe empirically that the Deep Equilibrium approach to training achieves competitive reconstruction quality and increased flexibility with respect to allocating computation budget at inference time. Recent work in deep inversion has questioned these methods' robustness to noise and unexpected inputs \cite{antun2020instabilities, raj2020improving, genzel2020solving}.

To examine whether the Deep Equilibrium approach is brittle to simple changes in the noise distribution, we varied the level of Gaussian noise added to the observations at test time and observed the effect on reconstruction quality in a setting where {\sc DE-Prox} and {\sc DU-Prox} perform similarly. Fig. \ref{fig:noisesensitivity} demonstrates that the Deep Equilibrium model {\sc DE-Prox} is more robust to variation in the noise level than the analogous Deep Unrolled approach {\sc DU-Prox}. The forward model used in Fig. \ref{fig:noisesensitivity} is $8\times$ MRI reconstruction.

\begin{figure}[ht]
\begin{center}
\subfigure{\label{fig:initializationstudy}\includegraphics[width=0.4\linewidth]{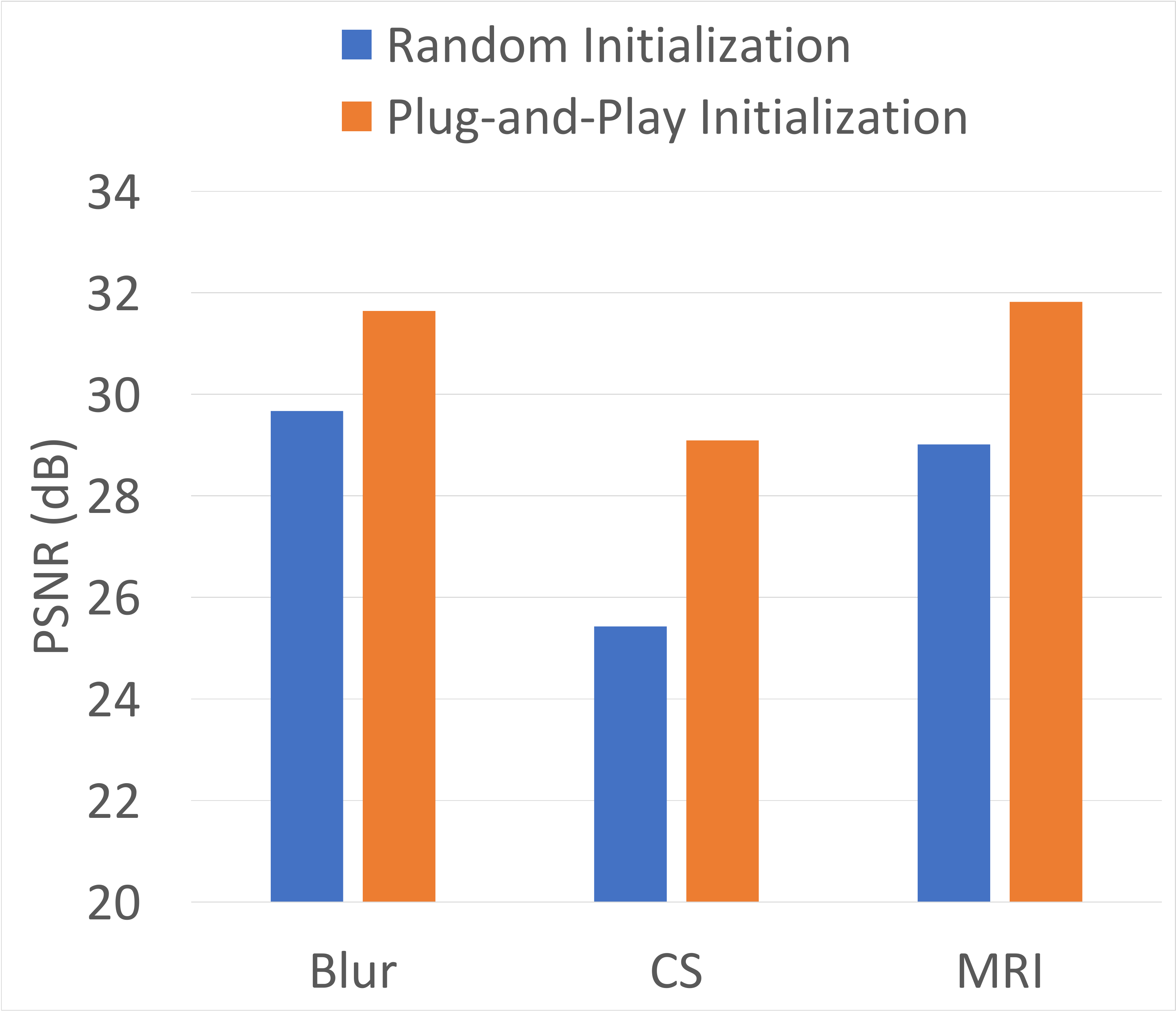}}
\subfigure{\label{fig:noisesensitivity}\includegraphics[width=0.4\linewidth]{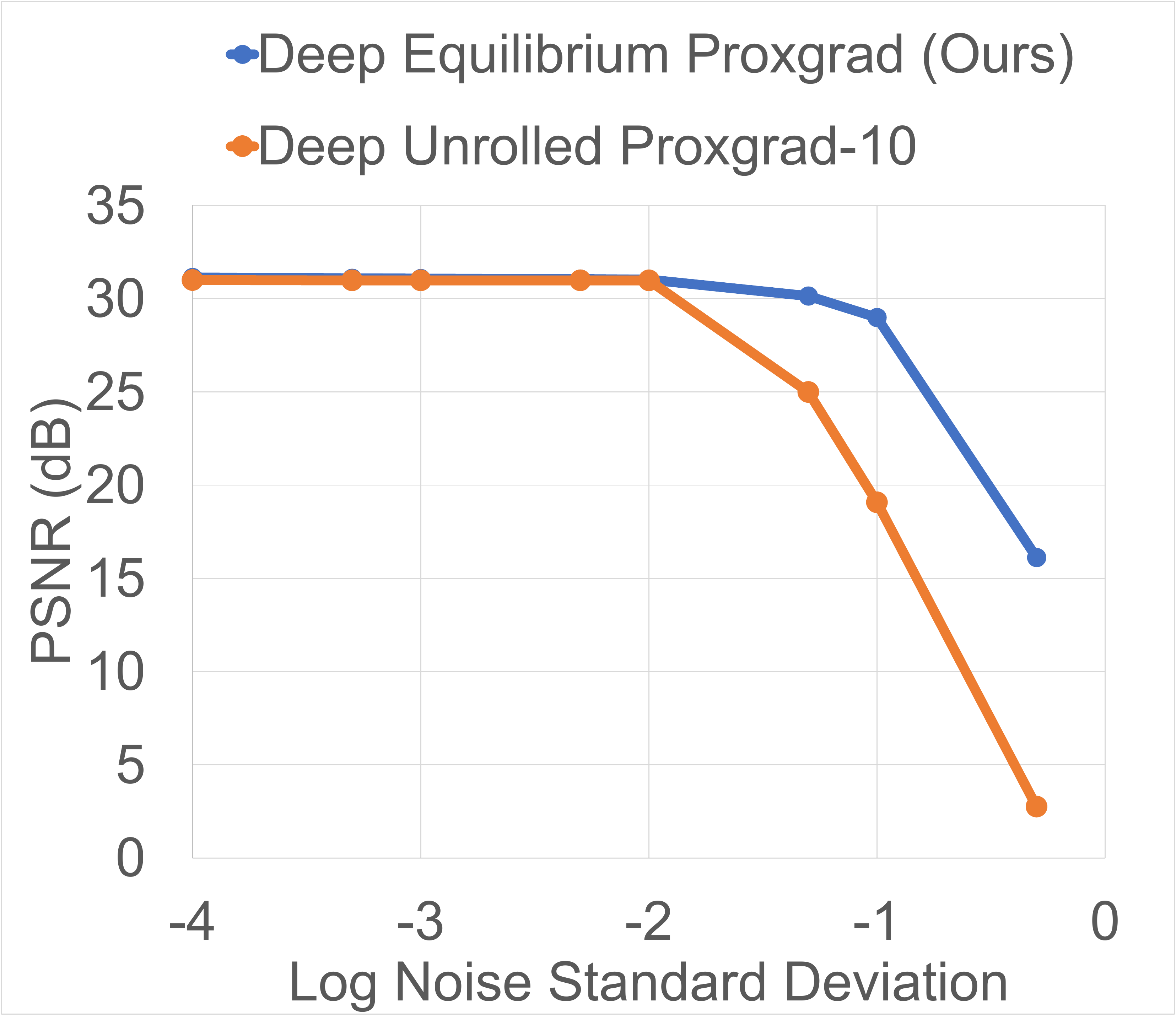}}
\caption{(a) Comparison of learned {\sc DE-Prox} reconstruction quality across three different inverse problems: Deblurring (Blur), compressed sensing (CS), and undersampled MRI reconstruction (MRI). In our experiements, initializing with a pretrained denoiser routinely offered as good or better reconstruction quality (in terms of PSNR) than a random initialization. (b) Noise sensitivity comparison between {\sc DU-Prox} and {\sc DE-Prox}. The forward model used is $8\times$ MRI reconstruction, and $\sigma$ here corresponds to the level of Gaussian noise added to observations.}
\end{center}
\end{figure}

\section{Conclusions}
This paper illustrates non-trivial quantitative benefits to using implicitly-defined infinite-depth networks for solving linear inverse problems in imaging. These empirical benefits complement convergence guarantees that are unavailable to widely-used deep unrolling methods. Other recent work has focused on such {\em implicit networks} akin to the deep equilibrium models considered here (e.g. \cite{el2019implicit}). Whether these models could lead to additional advances in image reconstruction remains an open question for future work. 
Furthermore, while the exposition in this work focused on \emph{linear} inverse problems, nonlinear inverse problems may be solved with iterative approaches just as well. The conditions under which deep equilibrium methods proposed here may be used on such iterative approaches are an active area of investigation.

\bibliographystyle{IEEEtran}
\bibliography{refs}
\newpage

\section{Appendix}

\subsection{Further Qualitative Results}

In this section, we provide further visualizations of the reconstructions produced by Deep Equilibrium models and the corresponding Deep Unrolled approaches, beyond those shown in the main body. Figures \ref{fig:deblurringsamples}, \ref{fig:cssamples}, and \ref{fig:mrisamples} are best viewed electronically, and contain the ground-truth images, the measurements (projected back to image space in the case of MRI and compressed sensing), and reconstructions by {\sc DU-Prox} and {\sc DE-Prox}.

We also visualize the intermediate iterates in the fixed-point iterations, to further demonstrate the convergence properties of DEMs for image reconstruction. We find that DEMs converge quickly to reasonable reconstructions, and maintain high-quality reconstructions after more than one hundred iterations. 

\begin{figure}[ht!]
\centering
\begin{tabular}{@{}c@{}c@{}c@{}c@{}c@{}}
\begin{minipage}{0.2\linewidth} Ground \\ Truth\end{minipage} & \includegraphics[width = 0.18\linewidth, align=c]{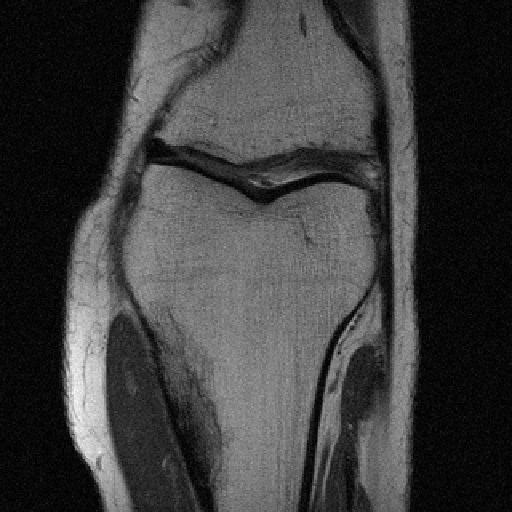} &
\includegraphics[width = 0.18\linewidth, align=c]{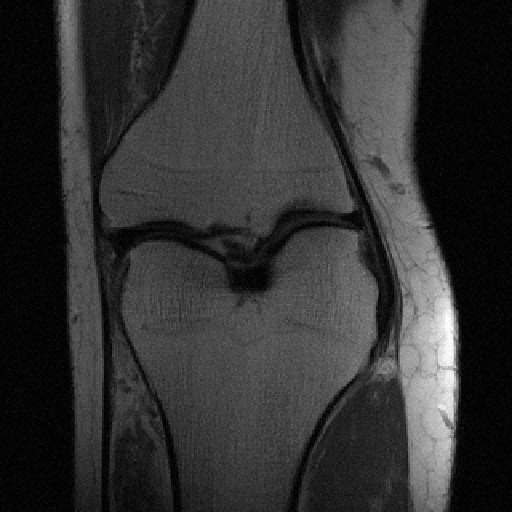} &
\includegraphics[width = 0.18\linewidth, align=c]{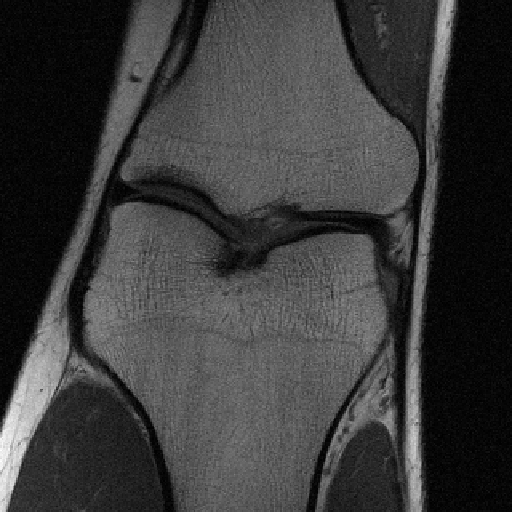} & 
\includegraphics[width = 0.18\linewidth, align=c]{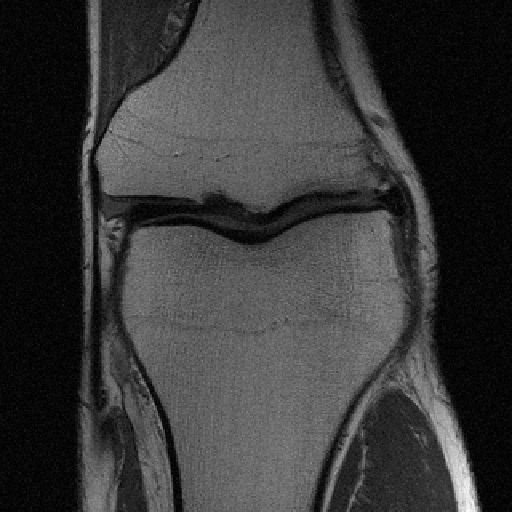} \\
\begin{minipage}{0.2\linewidth} IFFT \\ $A^\top y$ \end{minipage} & \includegraphics[width = 0.18\linewidth, align=c]{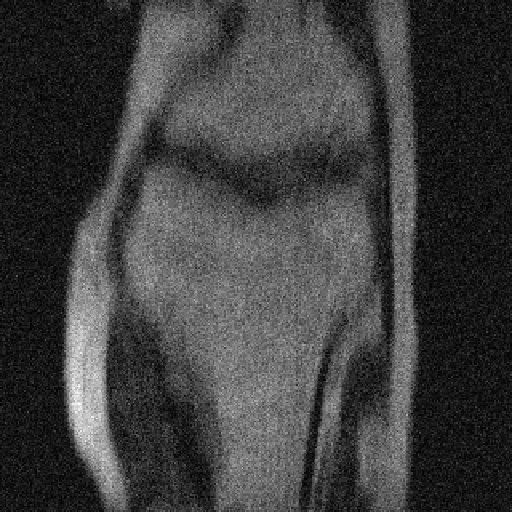} &
\includegraphics[width = 0.18\linewidth, align=c]{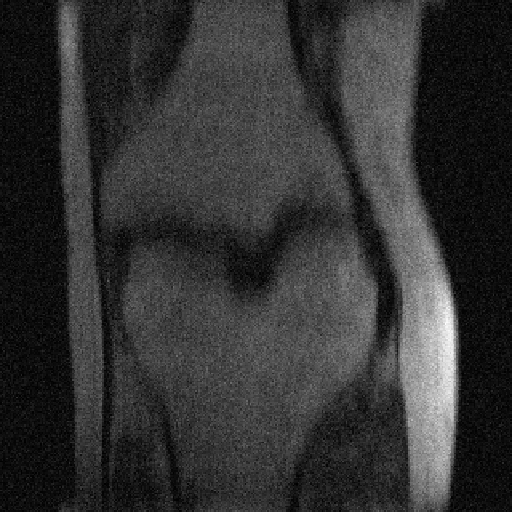} &
\includegraphics[width = 0.18\linewidth, align=c]{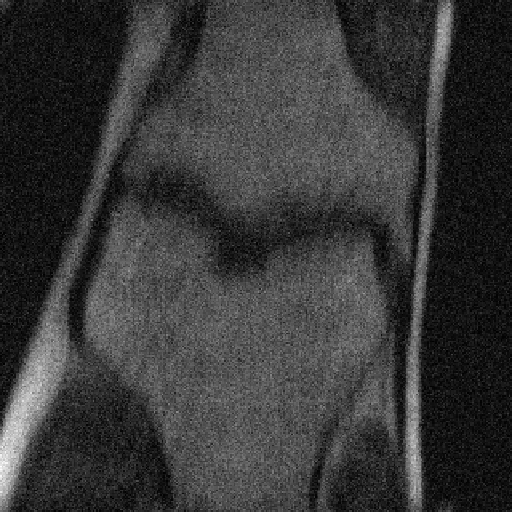} & 
\includegraphics[width = 0.18\linewidth, align=c]{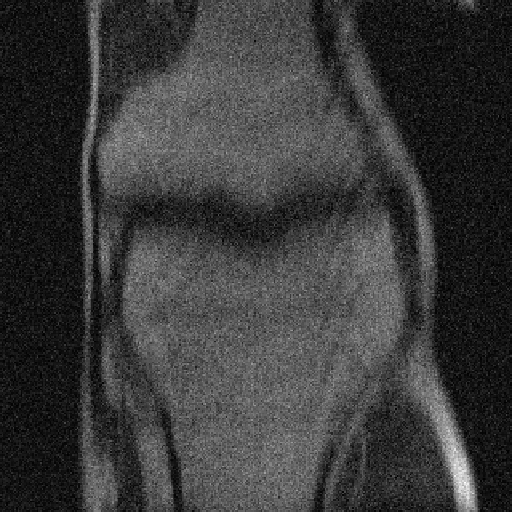} \\
\begin{minipage}{0.2\linewidth} {\sc DU-Prox} \end{minipage} & \includegraphics[width = 0.18\linewidth, align=c]{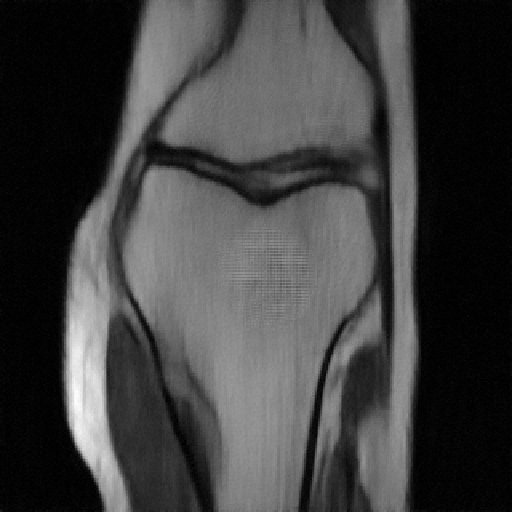} &
\includegraphics[width = 0.18\linewidth, align=c]{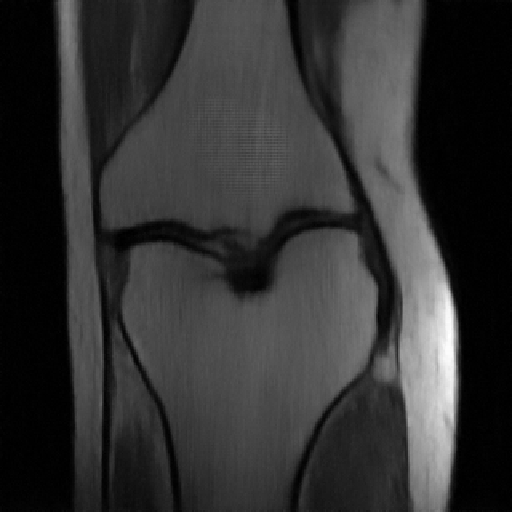} &
\includegraphics[width = 0.18\linewidth, align=c]{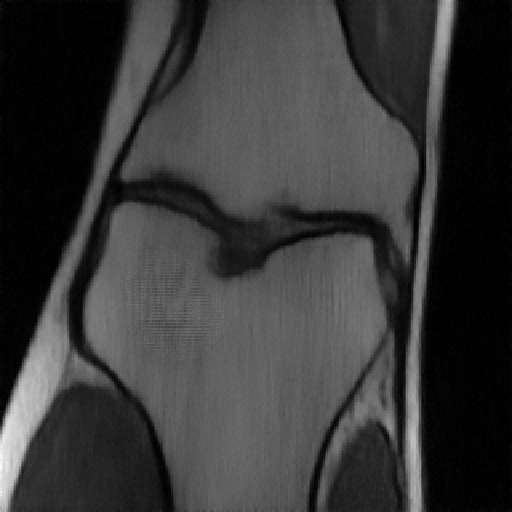} &
\includegraphics[width = 0.18\linewidth, align=c]{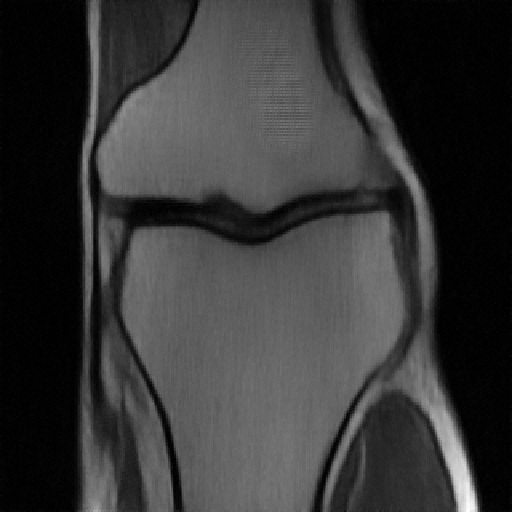} \\
\begin{minipage}{0.2\linewidth} {\sc DE-Prox} \\ (Ours)\end{minipage} & \includegraphics[width = 0.18\linewidth, align=c]{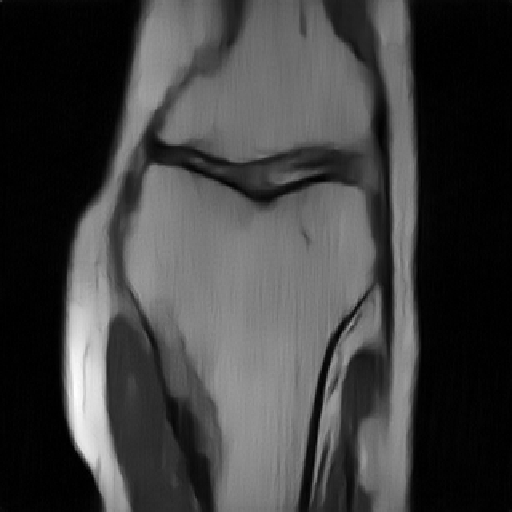} &
\includegraphics[width = 0.18\linewidth, align=c]{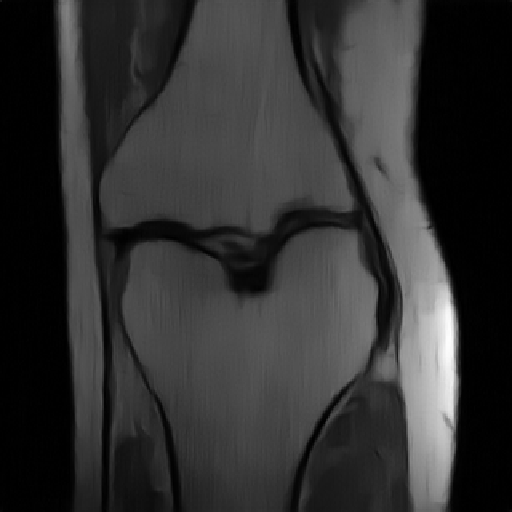} &
\includegraphics[width = 0.18\linewidth, align=c]{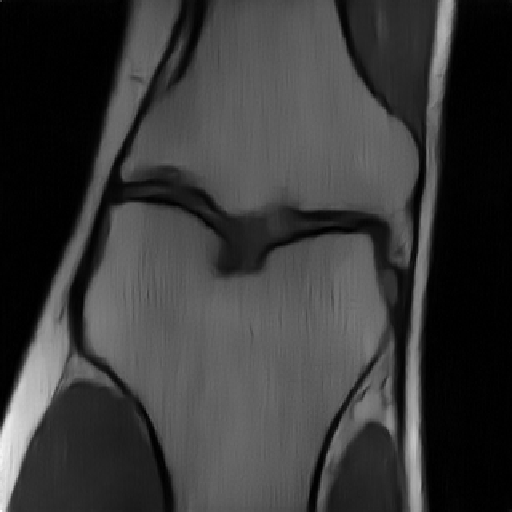} & 
\includegraphics[width = 0.18\linewidth, align=c]{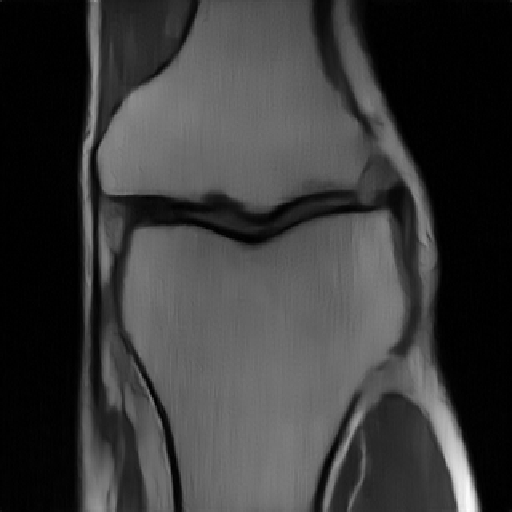}
\end{tabular}
\caption{Sample images and reconstructions with for $8\times$ accelerated MRI reconstruction with additive noise of $\sigma=0.01$. Best viewed electronically.}
\label{fig:mrisamples}
\end{figure}

\begin{figure}[ht!]
\centering
\begin{tabular}{@{}c@{}c@{}c@{}c@{}c@{}}
\begin{minipage}{0.2\linewidth} Ground \\ Truth\end{minipage} & \includegraphics[width = 0.15\linewidth, align=c]{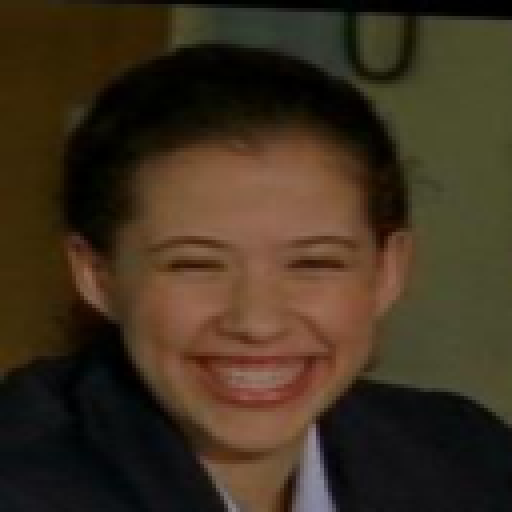} &
\includegraphics[width = 0.15\linewidth, align=c]{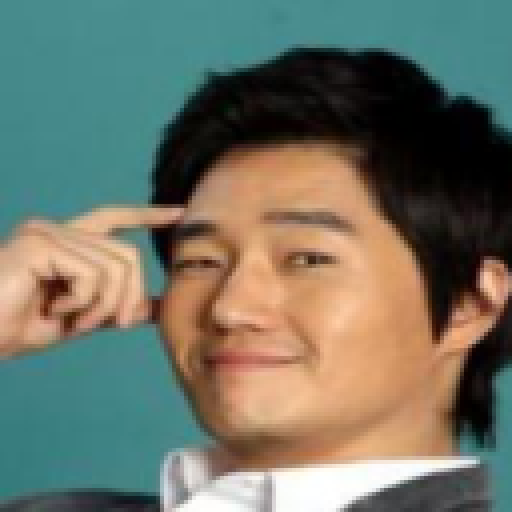} &
\includegraphics[width = 0.15\linewidth, align=c]{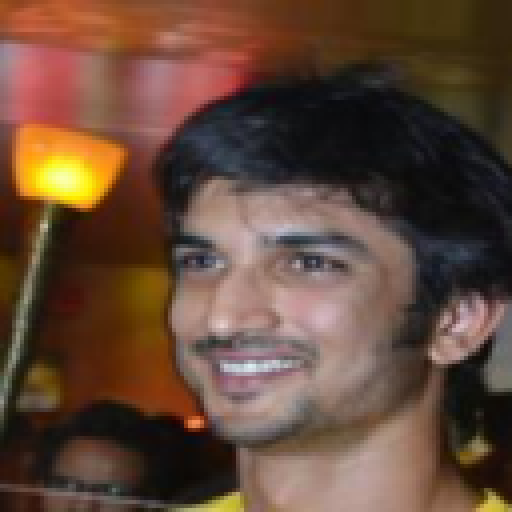} & 
\includegraphics[width = 0.15\linewidth, align=c]{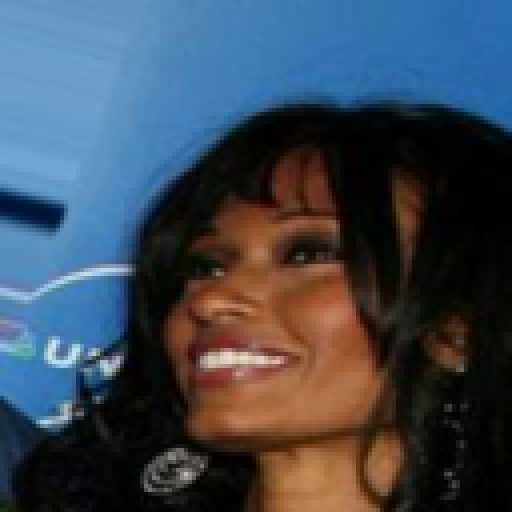} \\
\begin{minipage}{0.2\linewidth} $A^\top y$ \end{minipage} & \includegraphics[width = 0.15\linewidth, align=c]{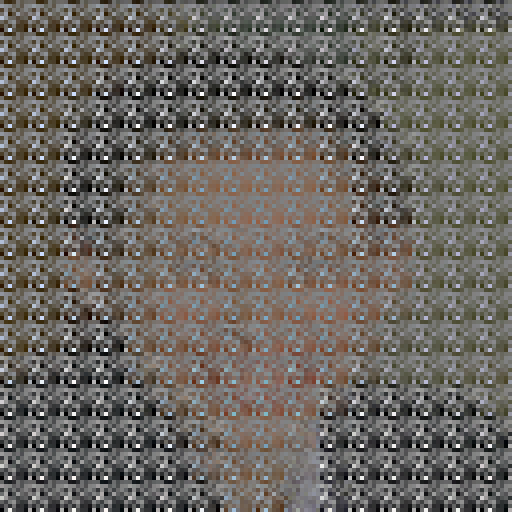} &
\includegraphics[width = 0.15\linewidth, align=c]{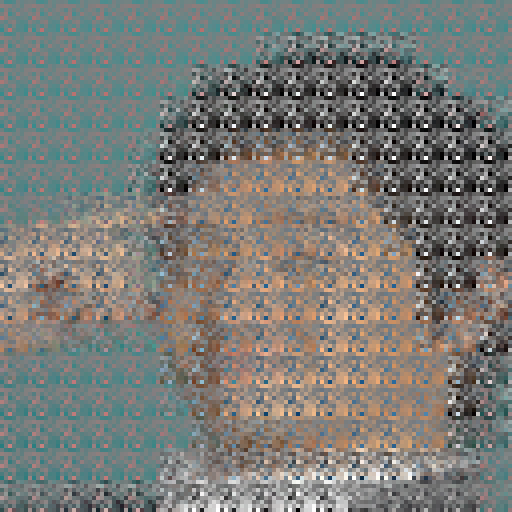} &
\includegraphics[width = 0.15\linewidth, align=c]{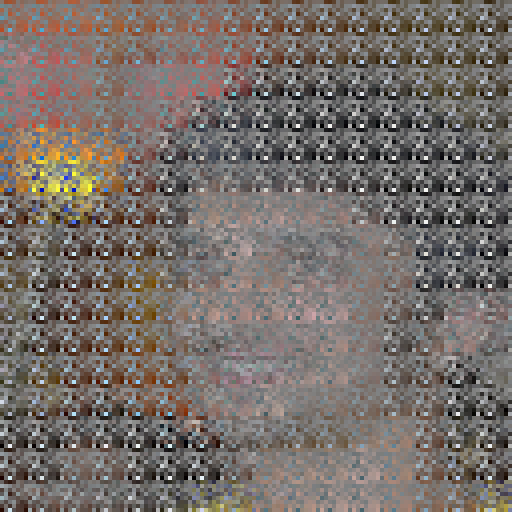} & 
\includegraphics[width = 0.15\linewidth, align=c]{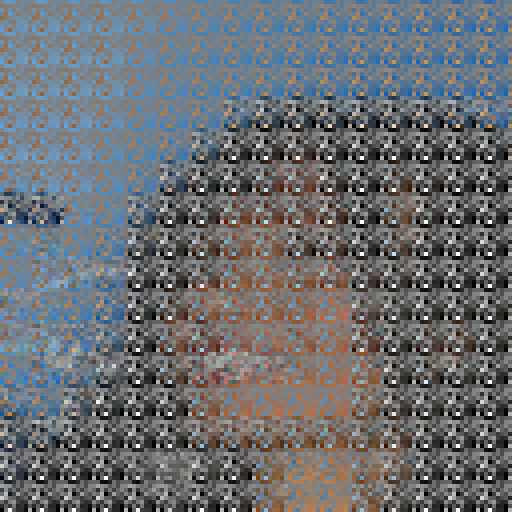} \\
\begin{minipage}{0.2\linewidth} {\sc DU-Prox} \end{minipage} & \includegraphics[width = 0.15\linewidth, align=c]{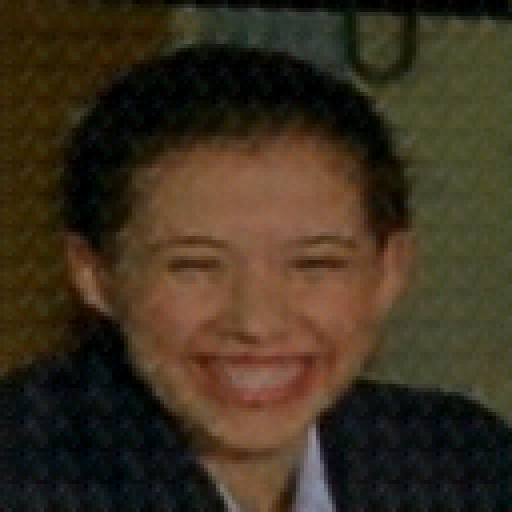} &
\includegraphics[width = 0.15\linewidth, align=c]{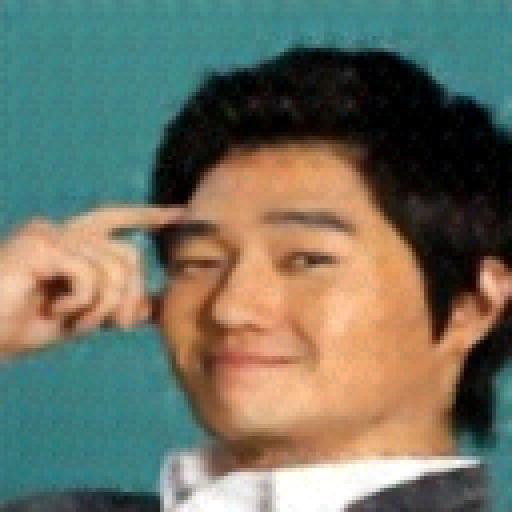} &
\includegraphics[width = 0.15\linewidth, align=c]{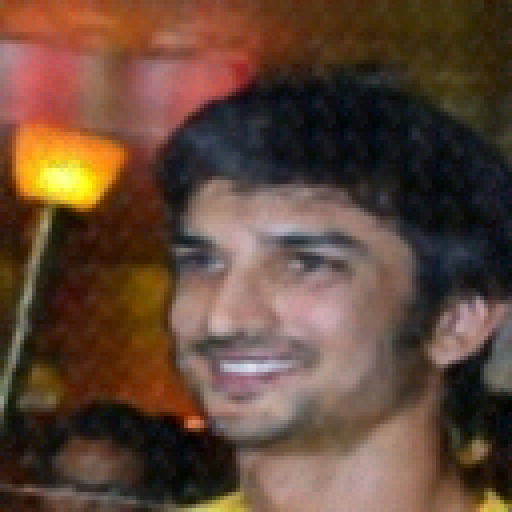} &
\includegraphics[width = 0.15\linewidth, align=c]{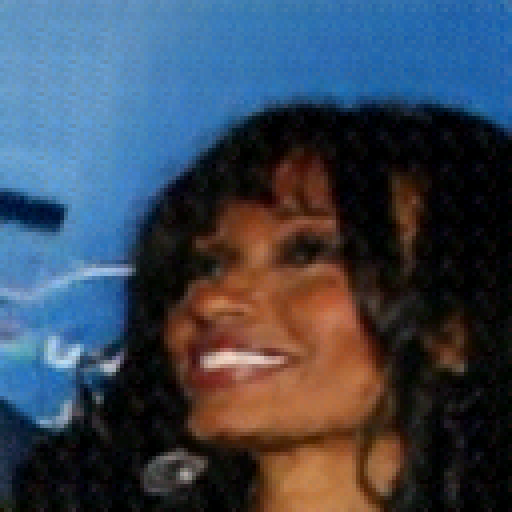} \\
\begin{minipage}{0.2\linewidth} {\sc DE-Prox} \\ (Ours)\end{minipage} & \includegraphics[width = 0.15\linewidth, align=c]{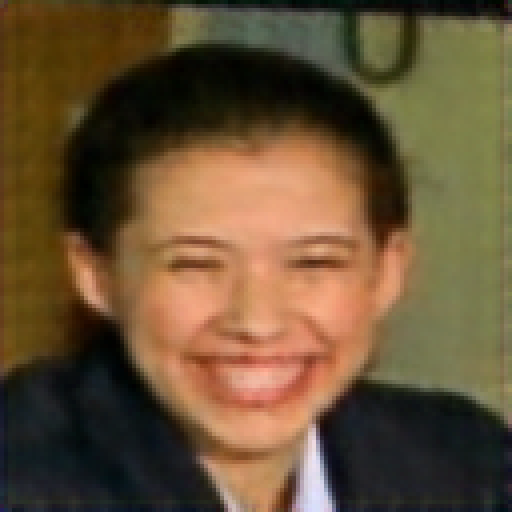} &
\includegraphics[width = 0.15\linewidth, align=c]{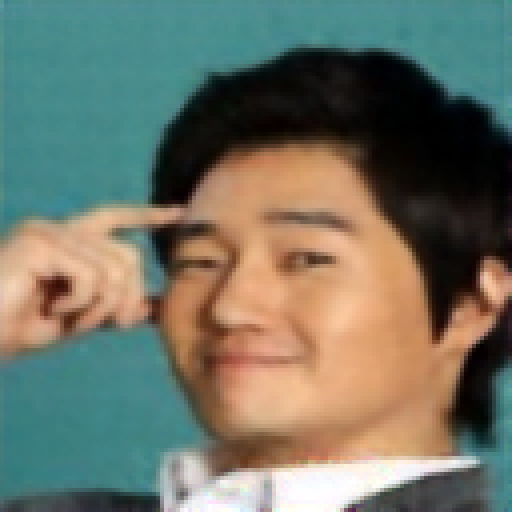} &
\includegraphics[width = 0.15\linewidth, align=c]{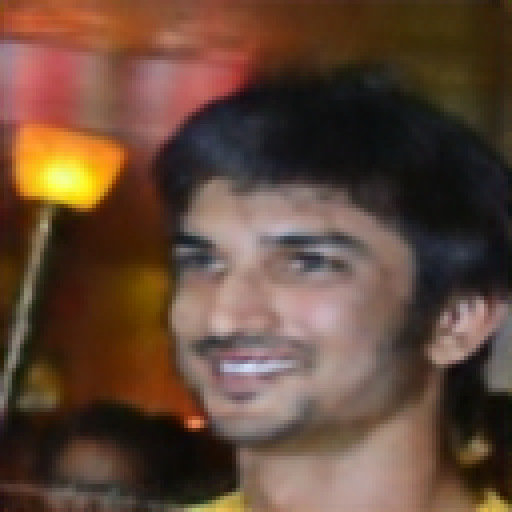} & 
\includegraphics[width = 0.15\linewidth, align=c]{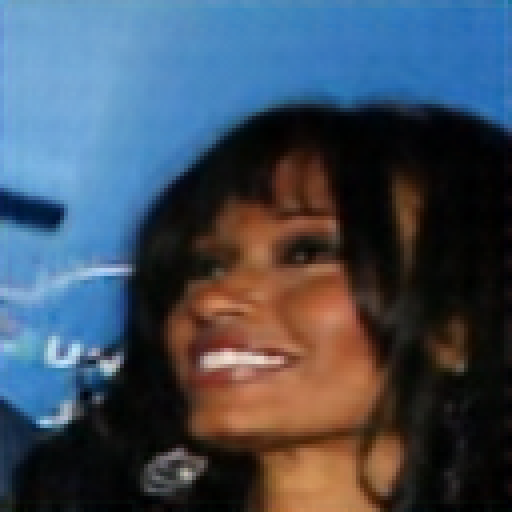}
\end{tabular}
\caption{Sample images and reconstructions for $4\times$ Gaussian compressed sensing with additive noise of $\sigma=0.01$. Best viewed electronically.}
\label{fig:cssamples}
\end{figure}

\begin{figure}[ht!]
\centering
\begin{tabular}{@{}c@{}c@{}c@{}c@{}c@{}}
\begin{minipage}{0.2\linewidth} Ground \\ Truth\end{minipage} & \includegraphics[width = 0.15\linewidth, align=c]{figures/supplement/deblur/truth_3.png} &
\includegraphics[width = 0.15\linewidth, align=c]{figures/supplement/deblur/truth_36.png} &
\includegraphics[width = 0.15\linewidth, align=c]{figures/supplement/deblur/truth_55.png} & 
\includegraphics[width = 0.15\linewidth, align=c]{figures/supplement/deblur/truth_93.png} \\
\begin{minipage}{0.2\linewidth} Measure-\\ments $y$ \end{minipage} & \includegraphics[width = 0.15\linewidth, align=c]{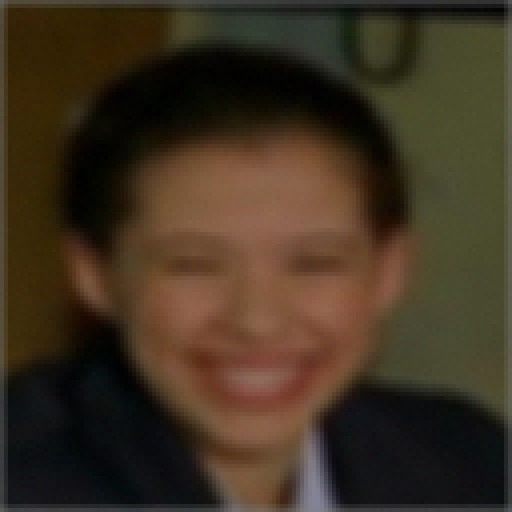} &
\includegraphics[width = 0.15\linewidth, align=c]{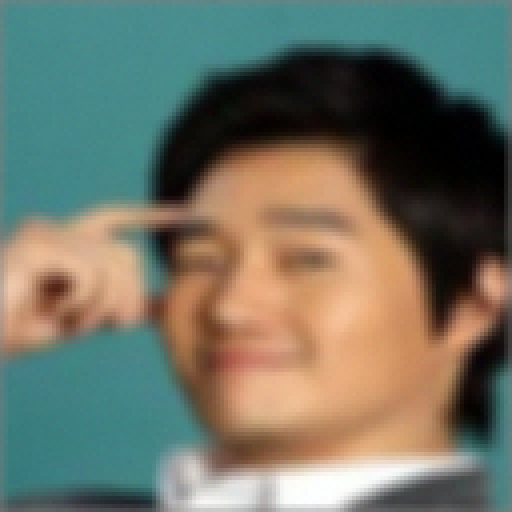} &
\includegraphics[width = 0.15\linewidth, align=c]{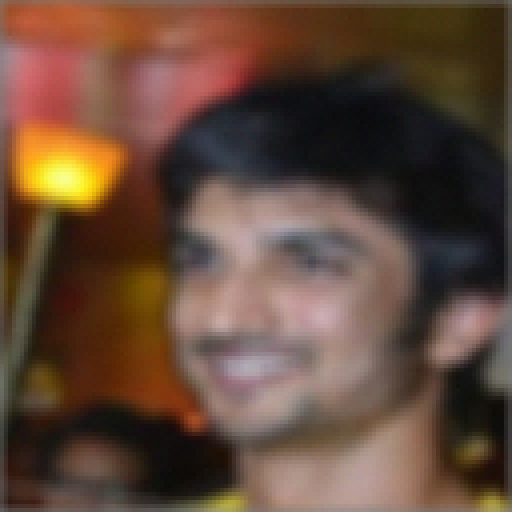} & 
\includegraphics[width = 0.15\linewidth, align=c]{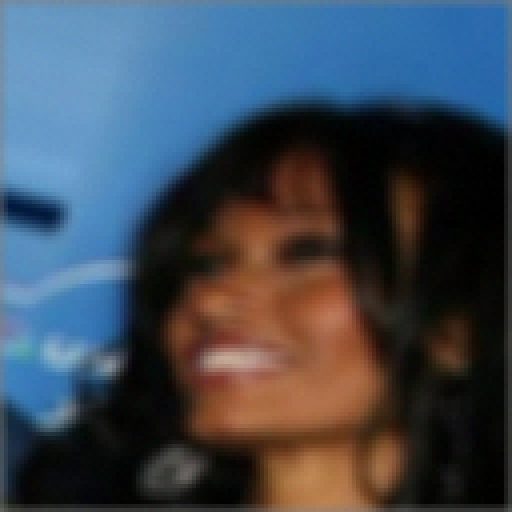} \\
\begin{minipage}{0.2\linewidth} {\sc DU-Prox} \end{minipage} & \includegraphics[width = 0.15\linewidth, align=c]{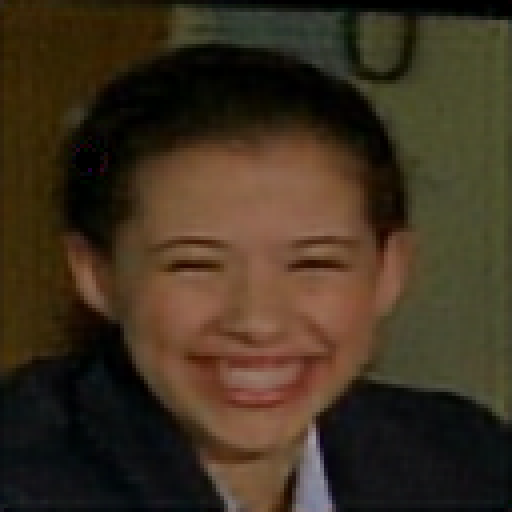} &
\includegraphics[width = 0.15\linewidth, align=c]{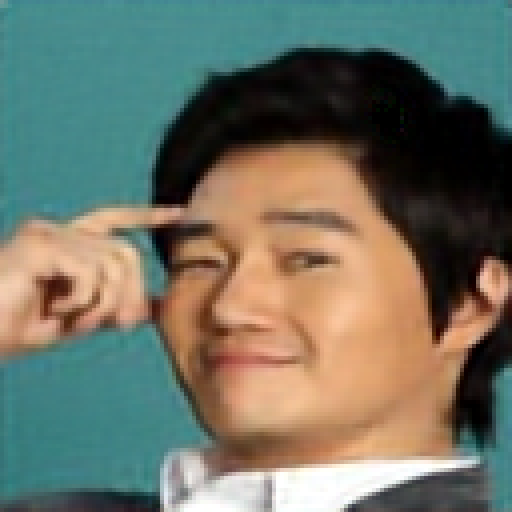} &
\includegraphics[width = 0.15\linewidth, align=c]{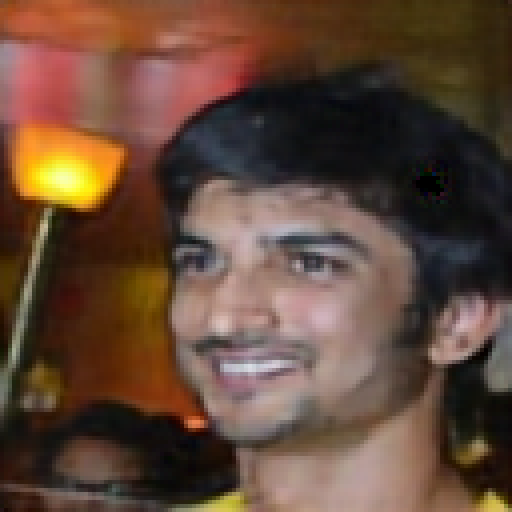} &
\includegraphics[width = 0.15\linewidth, align=c]{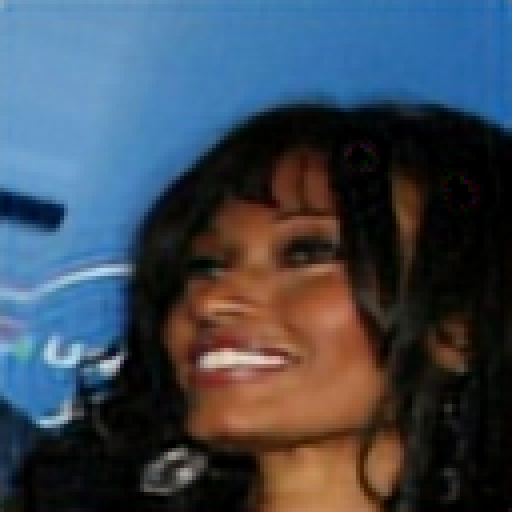} \\
\begin{minipage}{0.2\linewidth} {\sc DE-Prox} \\ (Ours)\end{minipage} & \includegraphics[width = 0.15\linewidth, align=c]{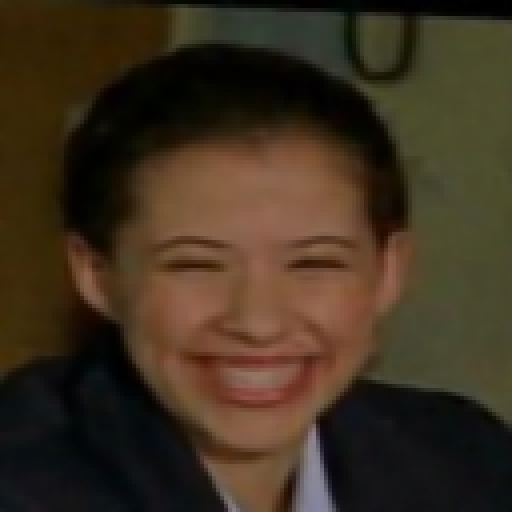} &
\includegraphics[width = 0.15\linewidth, align=c]{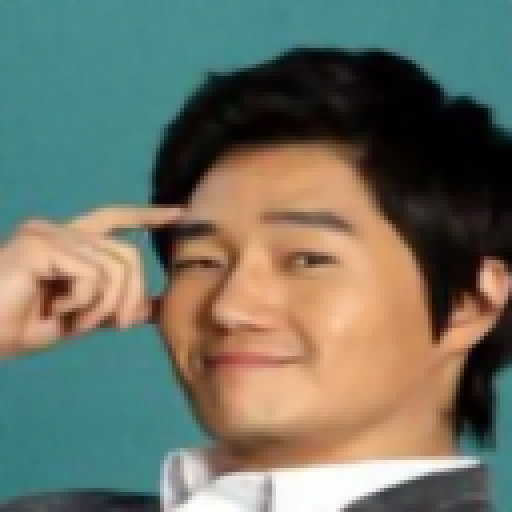} &
\includegraphics[width = 0.15\linewidth, align=c]{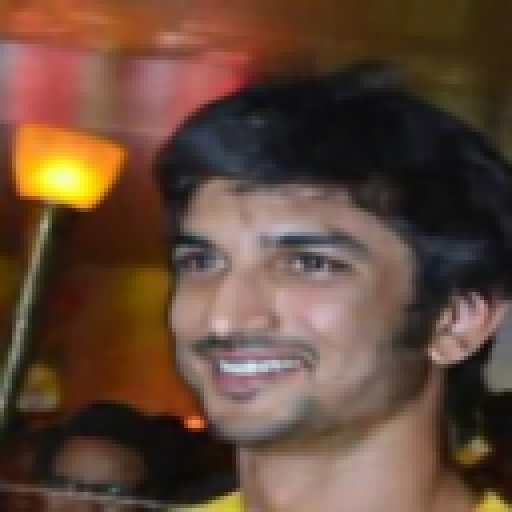} & 
\includegraphics[width = 0.15\linewidth, align=c]{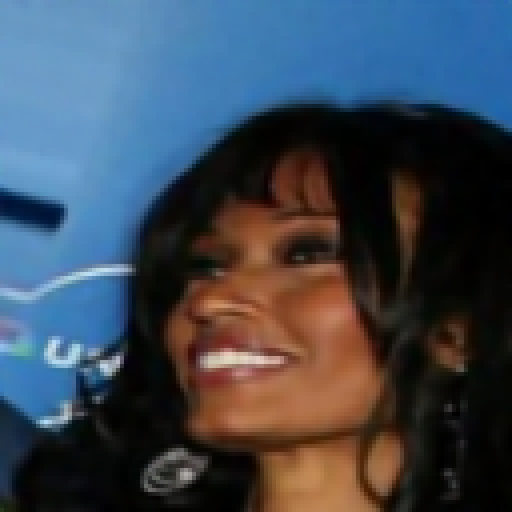}
\end{tabular}
\caption{Sample images and reconstructions for Gaussian deblurring with additive noise of $\sigma=0.01$. Best viewed electronically.}
\label{fig:deblurringsamples}
\end{figure}

\subsection{Visualizing Iterates}

In Figures \ref{fig:csiterates} and \ref{fig:mriiterates} we visualize the outputs of the $K$'th iteration of the mapping $f_\theta$ in {\sc DE-Prox}. We observe that across forward problems, the reconstructions converge to good reconstructions. 

We illustrate 90 iterations for compressed sensing and 31 for MRI reconstructions (as iterations terminate at 31 iterations). 

\begin{figure*}[ht!]
\renewcommand*{\arraystretch}{0}
\centering
\begin{tabular}{c@{}c@{}c@{}c@{}c@{}}
K=0 & K=10 & K=20 & K=30 & K=31 \vspace{2pt} \\
 \includegraphics[width = 0.15\linewidth, align=c]{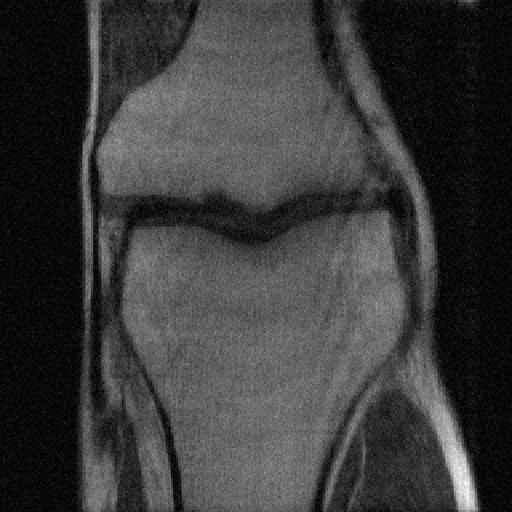} &
\includegraphics[width = 0.15\linewidth, align=c]{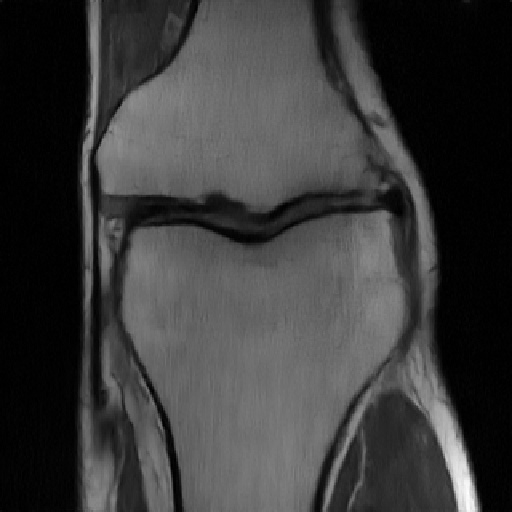} &
\includegraphics[width = 0.15\linewidth, align=c]{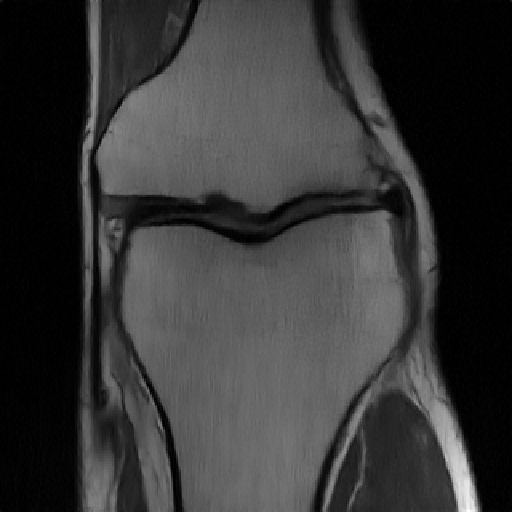} &
\includegraphics[width = 0.15\linewidth, align=c]{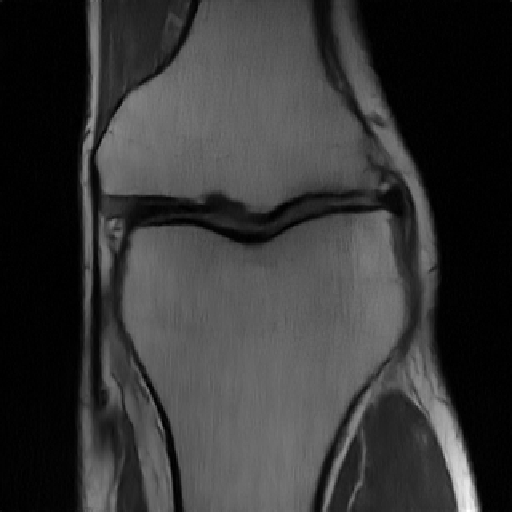} &
\includegraphics[width = 0.15\linewidth, align=c]{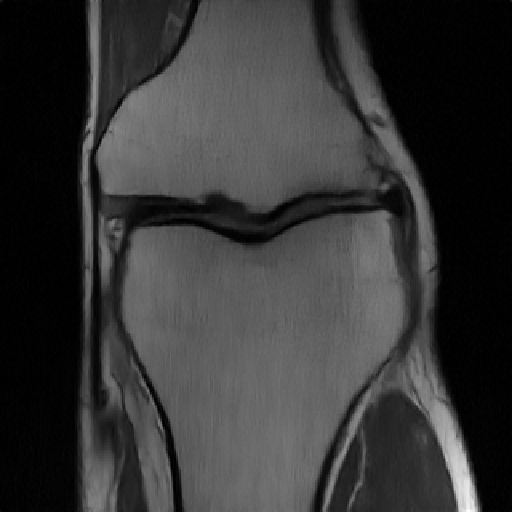} 
\end{tabular}
\caption{Sample images and reconstructions for MRI reconstruction with acceleration $4\times$ and additive Gaussian noise with $\sigma=0.01$. Each image represents the output of iterate number $K$. Below each image is the residual between iterate $K$ and the previously-visualized iterate, or in the case of $K=0$, between the input to the network and the output of the initial iterate. In this case, the algorithm stops iterating (as the relative norm between iterations drops below $10^{-3}$) at iteration 31. The ground truth may be viewed in the initial column of Figure \ref{fig:mrisamples}. Best viewed electronically.}
\label{fig:mriiterates}
\end{figure*}

\begin{figure*}[ht!]
\renewcommand*{\arraystretch}{0}
\centering
\begin{tabular}{c@{}c@{}c@{}c@{}c@{}}
K=0 & K=10 & K=20 & K=30 & K=40  \vspace{2pt} \\
 \includegraphics[width = 0.15\linewidth, align=c]{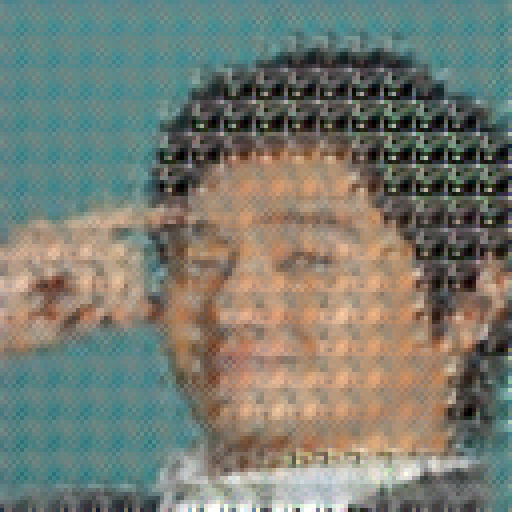} &
\includegraphics[width = 0.15\linewidth, align=c]{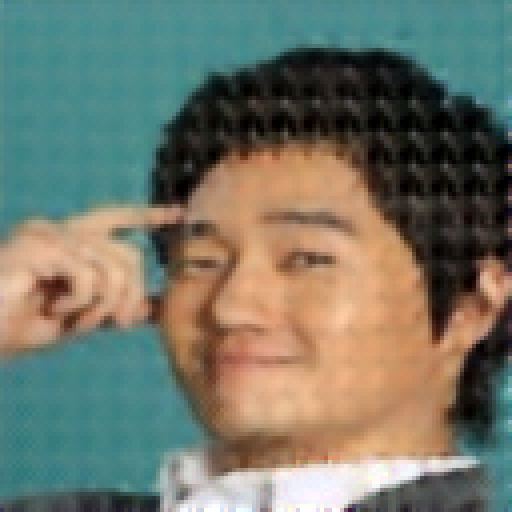} &
\includegraphics[width = 0.15\linewidth, align=c]{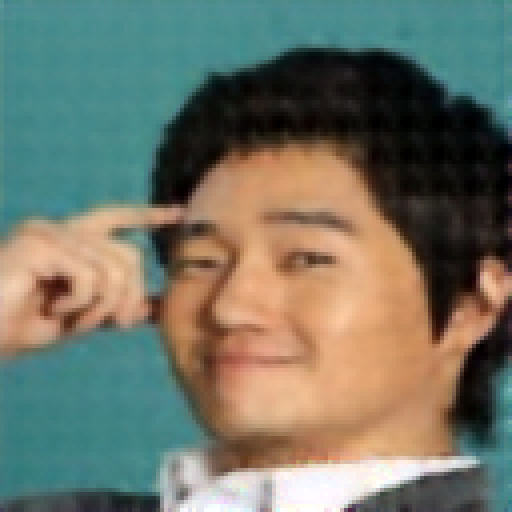} &
\includegraphics[width = 0.15\linewidth, align=c]{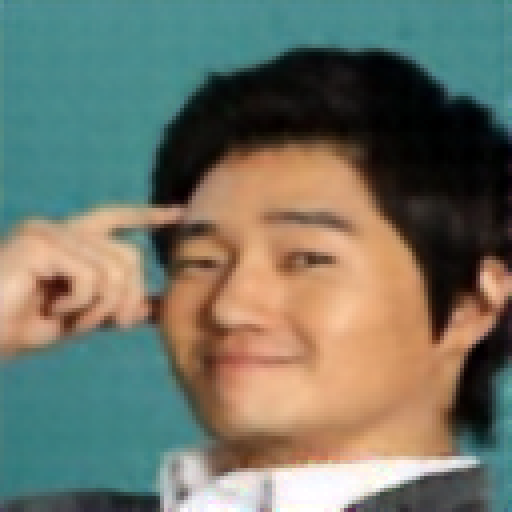} &
\includegraphics[width = 0.15\linewidth, align=c]{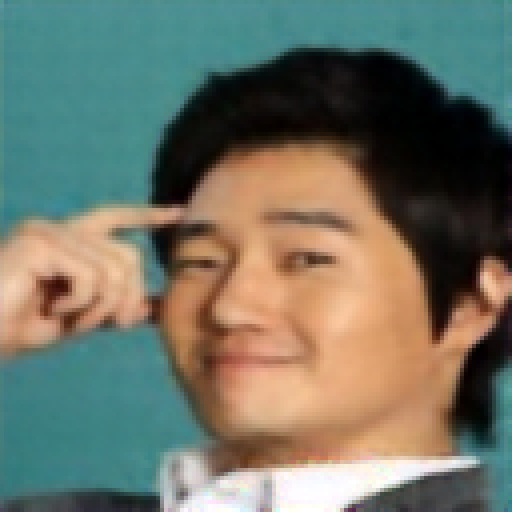} \\
\vspace{1pt} \\ 
K=50 & K=60 & K=70 & K=80 & K=90 \vspace{2pt} \\
\includegraphics[width = 0.15\linewidth, align=c]{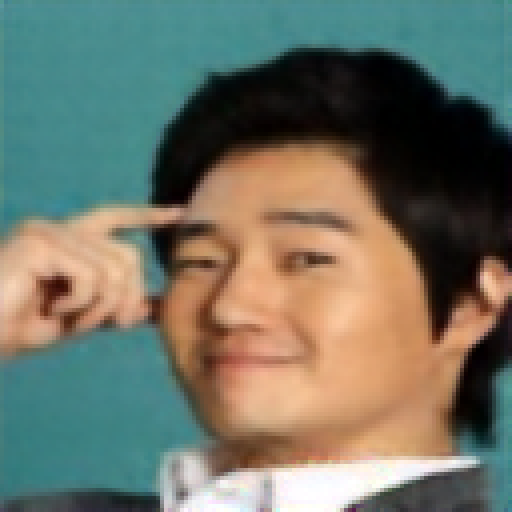} &
\includegraphics[width = 0.15\linewidth, align=c]{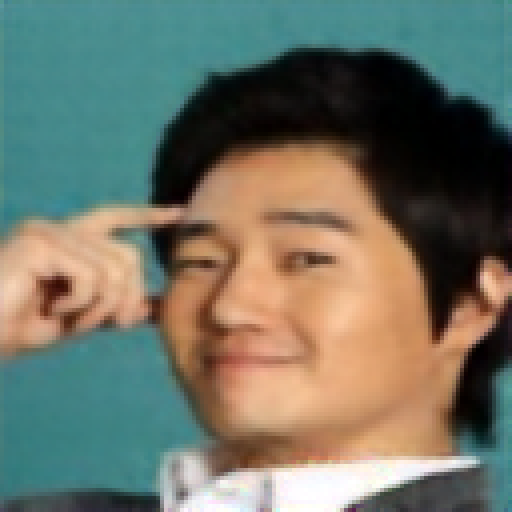} &
\includegraphics[width = 0.15\linewidth, align=c]{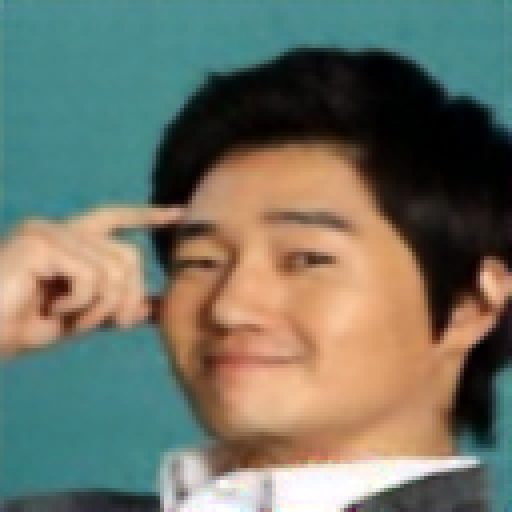} &
\includegraphics[width = 0.15\linewidth, align=c]{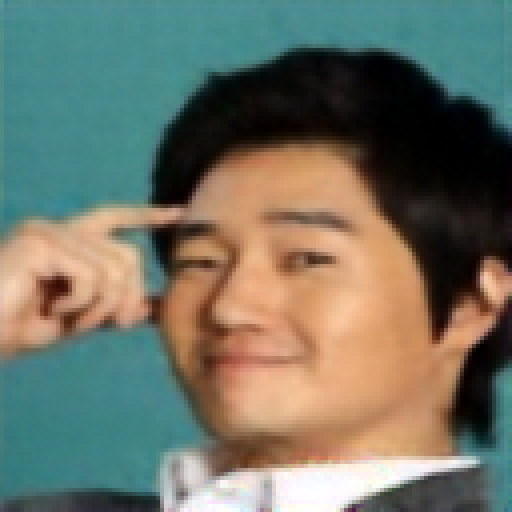} &
\includegraphics[width = 0.15\linewidth, align=c]{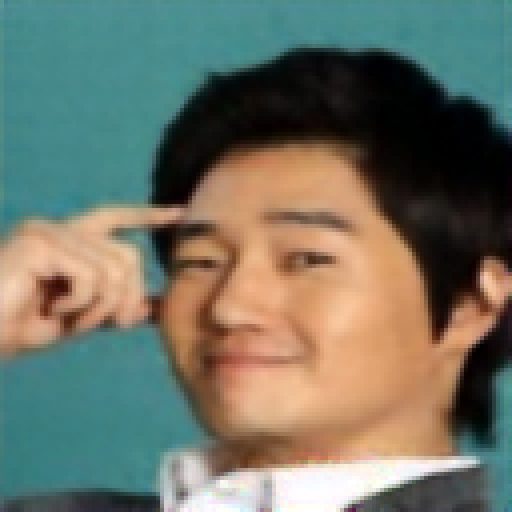}
\end{tabular}
\caption{Sample images and reconstructions from {\sc DE-Prox} reconstructions, with the forward model $4\times$ Gaussian compressed sensing with Gaussian noise with $\sigma=0.01$. Each image represents the output of iterate number $K$. The ground truth may be viewed in the final column of Figure \ref{fig:cssamples}. Best viewed electronically.}
\label{fig:csiterates}
\end{figure*}

\subsection{Further Experimental Details}

In this section we provide further details related to the experimental setup.

The input to the deblurring algorithms is the preconditioned measurement $(A^\top A + \lambda I)^{-1} A^\top y$, where $\lambda$ is set to be equal to the noise level $\sigma$. For MRI reconstruction and compressed sensing experiments, the input is instead simply $A^\top y$. The masks used in the MRI reconstruction experiments are based on a Cartesian sampling pattern, as in the standard fastMRI setting. For both $4\times$ and $8\times$, the center 4$\%$ of frequencies are fully sampled, and further frequencies are sampled according to a Gaussian distribution centered at 0 frequency with $\sigma=1$.

The compressed sensing design matrices have entries sampled and scaled so that each entry is drawn from a Gaussian distribution with variance $1/m$, where $A\in \mathcal{R}^{m\times n}$. The same design matrix is used for all learned methods.

Optimization algorithm parameters for RED, Plug-and-Play, and all Deep Equilibrium approaches are all chosen via a logarithmic grid search from $10^{-4}$ to $10^1$ with 20 elements in each dimension of the grid. All DU methods were trained for 10 iterations. All testing was done on an NVidia RTX 2080 Ti. All networks were trained on a cluster with a variety of computing resources \footnote{See: \url{https://slurm.ttic.edu/}}. Every experiment was run utilizing a single GPU-single CPU setup with less than 12 GB of GPU memory.

\end{document}